\documentclass[11pt]{article}

\usepackage[
  bookmarks=true,
  bookmarksnumbered=true,
  bookmarksopen=true,
  pdfborder={0 0 0},
  breaklinks=true,
  colorlinks=true,
  linkcolor=black,
  citecolor=black,
  filecolor=black,
  urlcolor=black,
]{hyperref}

\usepackage[small,compact]{titlesec}
\usepackage{times}
\usepackage{amsthm}
\usepackage{subfigure}
\usepackage{amsmath}
\usepackage{amsfonts}
\usepackage{amssymb}
\usepackage{graphicx}
\usepackage{fullpage}
\usepackage{pdfpages}
\usepackage{stmaryrd}
\usepackage{enumitem}
\usepackage[capitalise]{cleveref}
\setlist{nolistsep}

\newcommand{\ls}[1]
   {\dimen0=\fontdimen6\the\font \lineskip=#1\dimen0
\advance\lineskip.5\fontdimen5\the\font \advance\lineskip-\dimen0
\lineskiplimit=.9\lineskip \baselineskip=\lineskip
\advance\baselineskip\dimen0 \normallineskip\lineskip
\normallineskiplimit\lineskiplimit \normalbaselineskip\baselineskip
\ignorespaces }

\ls{1.2}

\newtheorem{dfn}{Definition}
\newtheorem{definition}{Definition}

\newtheorem{cor}{Corollary}
\newtheorem{corollary}{Corollary}
\newtheorem{lemma}{Lemma}

\newtheorem{theorem}{Theorem}

\newtheorem{remark}{Remark}

\newlist{parts}{enumerate}{1}
\crefname{partsi}{Part}{Parts}
\setlist[parts,1]{label=\arabic*.,ref=\arabic*}

\newcommand{\knownf}[1]{\ensuremath{F\node{#1}}}
\newcommand{\minval}[1]{\ensuremath{\mathit{Min}\node{#1}}}
\newcommand{\Majvals}[1]{\ensuremath{\mathit{Maj}\node{#1}}}
\newcommand{\knownvals}[1]{\ensuremath{\mathit{Vals}\node{#1}}}
\newcommand{\knownlows}[1]{\ensuremath{\mathit{Low}\node{#1}}}
\newcommand{\eqdef}{\triangleq}
\newcommand{\OptMaj}{\mbox{$\mbox{\sc Opt}_{\Maj}$}}
\newcommand{\OptMin}{\mbox{$\mbox{\sc Opt}_{\min}$}}
\newcommand{\OptMink}{\mbox{$\mbox{\sc Opt}_{\min\!\mbox{-}k}$}}

\newcommand{\set}[1]{\{#1\}}

\newcommand{\Crash}{\mathsf{Crash}}
\def\emptyset{\mbox{\O}}

\newcommand{\defemph}[1]{\textbf{\textit{#1}}}
\newcommand{\sat}{\models}

\newcommand{\decideZ}{\mathsf{decide(0)}}
\newcommand{\decideO}{\mathsf{decide(1)}}
\newcommand{\decide}{\mathsf{decide}}

\newcommand{\gammacr}{\gamma^{\mathrm{cr}}}
\newcommand{\Proc}{\mathsf{Procs}}
\newcommand{\tee}{\,\defemph{t}}

\newcommand{\Pz}{P_0}
\newcommand{\UPz}{\mbox{$\mbox{{\sc u-}}P_0$}}
\newcommand{\Maj}{\mathsf{Maj}}
\newcommand{\OptO}{\mbox{$\mbox{{\sc Opt}}_1$}}
\newcommand{\OptZ}{\mbox{$\mbox{{\sc Opt}}_0$}}
\newcommand{\UOptZ}{\mbox{$\mbox{{\sc u-Opt}}_0$}}
\newcommand{\UOptMink}{\mbox{$\mbox{\sc u-Prot}_{\min\!\mbox{-}k}$}}

\newcommand{\Fmodel}{{\cal F}}
\newcommand{\CG}{{\cal G}}
\newcommand{\Vals}{{\tt V}}
\newcommand{\Vecs}{\vec{\Vals}}
\newcommand{\dom}{\,{\preceq}\,}
\newcommand{\hmwopt}{P0_{\mathrm{opt}}}
\newcommand{\pdo}{unbeatable}
\newcommand{\FP}{\mathsf{F}}
\newcommand{\nnz}{\mathsf{never\hbox{-}known}(\exists 0)}
\newcommand{\node}[1]{\langle#1\rangle}
\newcommand{\Ga}{\CG_\alpha}
\newcommand{\dec}{\mathsf{d}}
\newcommand{\cw}{\exists\mathsf{correct}(w)}
\newcommand{\cv}{\exists\mathsf{correct}(v)}
\newcommand{\cz}{\exists\mathsf{correct}(0)}

\DeclareMathOperator{\Bary}{\sf Bary}
\DeclareMathOperator{\Star}{\sf St}
\DeclareMathOperator{\Car}{\sf Car}
\DeclareMathOperator{\Div}{\sf Div}
\DeclareMathOperator{\bdry}{\sf Bd}

\newcommand{\ang}[1]{\langle{#1}\rangle}

\newcommand{\cK}{\ensuremath{\mathcal{K}}}
\newcommand{\cL}{\ensuremath{\mathcal{L}}}
\newcommand{\cP}{\ensuremath{\mathcal{P}}}

\newcommand{\kAgreement}{{\bf \defemph{k}-Agreement}}
\newcommand{\Agreement}{{\bf Agreement}}
\newcommand{\UnikAg}{{\bf Uniform \defemph{k}-Agreement}}
\newcommand{\UniAg}{{\bf Uniform Agreement}}
\newcommand{\Decision}{{\bf Decision}}
\newcommand{\Validity}{{\bf Validity}}

\hypersetup{
  pdfauthor      = {Armando Casta{\~n}eda <armando@cs.technion.ac.il>, Yannai A. Gonczarowski <yannai@gonch.name>, and Yoram Moses <moses@ee.technion.ac.il>},
  pdftitle       = {Good, Better,  Best! --- Unbeatable Protocols for Consensus and Set Consensus},
}

\begin{document}
 
\title{%
Good, Better,  Best! --- Unbeatable Protocols \\ for Consensus and Set Consensus%
\footnote{Part of the results of this paper were announced, without details and without proof, as a brief announcement in PODC 2013~\cite{AYY-PODC-BA}.}
}
\author{Armando Casta\~{n}eda\\Technion\\{\small armando@cs.technion.ac.il}
\and Yannai A.~Gonczarowski\\The Hebrew University of Jerusalem\\and Microsoft Research\\ {\small yannai@gonch.name}
\and Yoram Moses\\Technion\\{\small moses@ee.technion.ac.il}
}
\date{November 11, 2013}

\begin{titlepage}
\maketitle
\vspace{-2mm}
\begin{abstract}
While the very first consensus protocols for the synchronous model were designed to match the {\em worst-case} lower bound, deciding in exactly $\tee+1$ rounds in all runs, it was soon realized that they could be strictly improved upon by {\em early stopping} protocols. These dominate the first ones, by always deciding in at most~$\tee+1$ rounds, but often much faster. 
A protocol is \emph{unbeatable} if it can't be strictly dominated. Namely, if no protocol~$Q$ can decide  strictly earlier than~$P$ against at least one adversary strategy, while deciding at least as fast as~$P$ in all cases. 
Unbeatability is often a much more suitable notion of optimality for distributed protocols than worst-case performance. 
Halpern, Moses and Waarts in \cite{HalMoWa2001}, who introduced this notion, presented a general logic-based transformation of any consensus protocol to an unbeatable protocol that dominates it, and suggested a particular unbeatable consensus protocol. Their analysis is based on a notion of {\em continual common knowledge}, which is not easy to work with in practice. 
Using a more direct knowledge-based analysis, this paper studies unbeatability for both consensus and $k$-set consensus. 
We present  unbeatable solutions to 
\emph{non-uniform} consensus and $k$-set consensus, and \emph{uniform} consensus
in synchronous message-passing contexts with crash failures.
Our consensus protocol strictly dominates the one suggested in \cite{HalMoWa2001}, showing that their protocol {\em is} in fact beatable. 

The $k$-set consensus problem  is much more technically challenging than consensus, and its analysis has triggered the development of the topological approach to distributed computing. 
Worst-case lower bounds for this problem have required either techniques based on algebraic topology~\cite{GHP}, or reduction-based proofs~\cite{AGGT,GGP}. 
Our proof of unbeatability is purely combinatorial, and is a direct, albeit nontrivial, generalization of the one for consensus. We also present an alternative topological unbeatability proof 
that allows to understand the connection between the connectivity of protocol complexes
and the decision time of processes.
All of our protocols make use of a notion of a \defemph{hidden path} of nodes relative to a process~$i$ at time~$m$, in which a value unknown to~$i$ at~$m$ may be seen by others. This is a structure that can implicitly be found in lower bound proofs for consensus going back to the '80s~\cite{DS}. Its use in our protocols sheds light on the mathematical structure underlying the consensus problem and its variants. 
 
For the synchronous model, only solutions to the {\em uniform} variant of $k$-set consensus have been offered. Based on our unbeatable protocols for uniform consensus and for non-uniform $k$-set consensus, we present a uniform $k$-set consensus protocol that strictly dominates all known solutions to this problem in the synchronous model. 
\end{abstract}

\vfill
\vfill
\vfill
\pagenumbering{Roman}
\thispagestyle{empty}
\end{titlepage}
\pagenumbering{arabic}

\section{Introduction}

\vspace{-0.35cm}

Following \cite{HMT11}, we say that a protocol~$P$ is a \defemph{worst-case optimal} solution to a decision task~$S$ in a given model 
if it solves~$S$, and decisions in~$P$ are always taken no later than the {\em worst-case} lower bound for decisions in this problem. 
The very first consensus protocols were worst-case optimal, deciding in exactly $\tee+1$ rounds in all runs \cite{DS,PSL}.
It was soon realized, however, that they could be strictly improved upon by \defemph{early stopping} protocols~\cite{DRS}. 
The latter are also worst-case optimal, but they strictly improve upon the original ones because they can often decide much faster than the original ones. 
This paper is concerned with the study and construction of protocols that cannot be strictly improved upon, and are thus optimal in a much stronger sense. 

In benign failure models it is typically possible to define the behaviour of the environment (i.e., the adversary) in a manner that is independent of the protocol, in terms of a pair $\alpha=(\vec{v},\FP)$ consisting of a vector $\vec{v}$ 
of initial values and a failure pattern~$\FP$. (A formal definition is given in Section~\ref{sec:model}.)
A failure model~$\Fmodel$  is identified with a set of (possible) failure patterns. 
For ease of exposition, we will think of such a pair $\alpha=(\vec{v},\FP)$ as a particular {\em adversary}. 
A deterministic protocol~$P$ and an adversary~$\alpha$ uniquely define a run $r=P[\alpha]$. 
With this terminology, we can compare the performance of different decision protocols solving a particular task in a given context $\gamma=(\Vecs,\Fmodel)$, where $\Vecs$ is a set of initial vectors. 
A decision protocol $Q$ \defemph{dominates} a protocol~$P$ in~$\gamma$, denoted by $Q\boldsymbol{\dom_\gamma} P$ if, for all adversaries $\alpha$ and every process~$i$, if $i$ decides in~$P[\alpha]$ at time $m_i$, then $i$ decides in $Q[\alpha]$ at some time $m'_i\le m_i$. Moreover, we say that $Q$  \defemph{strictly dominates} $P$
if $Q\dom_\gamma P$ and  $P\!\!\boldsymbol{\not}\!\!\!\dom_\gamma Q$. I.e., if it dominates~$P$ and for some $\alpha\in\gamma$ there exists a process~$i$ that decides in $Q[\alpha]$ {\em strictly before} it does so in $P[\alpha]$. 
In the crash failure model, the early-stopping protocols of \cite{DRS} strictly dominate the original protocols of \cite{PSL}, which always decided at time $\tee+1$. Nevertheless, these early stopping protocols may not be optimal solutions to consensus. 
Following \cite{HMT11} a protocol~$P$ is said to be an  \defemph{all-case optimal} solution to a decision task~$S$ in a context~$\gamma$ if it solves~$S$ and, moreover, $P$ dominates every protocol~$P'$ that solves~$S$ in~$\gamma$. 
Dwork and Moses presented all-case optimal solutions to the {\em simultaneous} variant of consensus, in which all decisions are required to occur at the same time~\cite{DM}.
 For the standard ({\em eventual}) variant of consensus, in which decisions are not required to occur simultaneously, Moses and Tuttle showed that no all-case optimal solution exists~\cite{MT}. 
Consequently, Halpern, Moses and Waarts in \cite{HalMoWa2001} initiated the study of a notion of optimality 
that is achievable by eventual consensus protocols:

\begin{definition}[Halpern, Moses and Waarts]
A protocol $P$ is an \defemph{unbeatable} solution to a decision task~$S$ in a context~$\gamma$ if $P$ solves~$S$ in~$\gamma$ and no protocol $Q$ solving~$S$ in~$\gamma$ strictly dominates~$P$.%
\end{definition}

\vspace{-0.35cm}

Thus, $P$ is unbeatable if for all protocols~$Q$ that solve~$S$, if there exist an adversary~$\alpha$ and process~$i$ such that~$i$ 
decides in $Q[\alpha]$ strictly earlier than it does in $P[\alpha]$, then there must exist some adversary~$\beta$ and process~$j$ such that $j$ decides strictly earlier in $P[\beta]$ than it does in $Q[\beta]$. An unbeatable solution for~$S$ is $\dom$-minimal among the solutions of~$S$.\footnote{All-case optimal protocols are called {\em ``optimal in all runs''} in \cite{DM}. 
They are called  {\em ``optim\defemph{um}''} protocols by Halpern, Moses and Waarts in \cite{HalMoWa2001}, and unbeatable ones are simply called {\em ``optim\defemph{al}''} there.}

Halpern, Moses and Waarts 
observed that for every consensus protocol $P$ there exists an unbeatable protocol $Q_P$ that dominates~$P$. Moreover, they showed a two-step transformation that defines such a protocol~$Q_P$ based on~$P$. This transformation is based on a notion of {\em continual} common knowledge that is computable, but not in a computationally-efficient manner. They also present a simple and efficient consensus protocol $\hmwopt$ that is claimed to be unbeatable in the crash failure model.

This paper is concerned with the construction of concrete unbeatable~protocols for a number of variants of consensus
in synchronous, message-passing systems with crash failures. 
A new knowledge-based analysis \cite{FHMV,HM1} allows a simpler and more intuitive 
approach to unbeatability than that used in~\cite{HalMoWa2001}. 
Our main contributions are:
\vspace{-0.2cm}

\begin{enumerate}

\item A knowledge-based approach to the design and presentation of consensus protocols is employed, based on the knowledge of preconditions principle. 

\item The first unbeatable protocols are presented for (non-uniform) consensus and $k$-set consensus, and uniform consensus in the crash failure model. A protocol that strictly dominates all known protocols for uniform $k$-set consensus is presented. 

\item The unbeatable consensus protocol strictly dominates the $\hmwopt$ protocol from \cite{HalMoWa2001}, proving that $\hmwopt$, which was claimed to be beatable is {\em not} unbeatable. 

\item The proofs of unbeatability are combinatorial, and do not require topological or reduction-based arguments 
even for the $k$-set consensus protocol. 
A second, topological, proof for the $k$-set consensus protocol is presented in the appendix, and is compared with the combinatorial proof.
This is the first result that we know to have proofs of both kinds, and the comparison sheds light on the relationship between these two approaches. 

\item While the proof for consensus is strikingly succinct, both the proofs for $k$-set consensus and for uniform consensus are technically challenging and highly nontrivial. 

\end{enumerate}

\vspace{-0.2cm}

Full proofs of all technical claims stated in the paper are given in the Appendix.
 
In the rest of this section we sketch the intuition behind, and the structure of, our unbeatable protocols for consensus and $k$-set consensus in the crash failure model. The technical development substantiating this sketch is presented in the later sections. 

Denote by $\exists v$ the fact that at least one of the processes started out with initial value~$v$. 
In the standard 
(non-uniform)
version of consensus, there is an \emph{a priori} bound of $\tee$ on the number of failures, 
initial values are $v_i\in\{0,1\}$, and the following properties must hold in every run~$r$: 

\vspace{-0.15cm}

\begin{itemize}
\item[]{\bf Agreement:}\quad All correct processes that decide in~$r$ must decide on the same value. 
\item[]{\bf Decision:}\quad Every correct process must decide on some value, and 
\item[]{\bf Validity:}\quad For every value~$v$, a decision on~$v$ is allowed only if~$\exists v$ holds. 
\end{itemize}

\vspace{-0.15cm}
Since $\exists v$ is a precondition for deciding~$v$ by the \Validity\ property, a process cannot decide~$v$ unless it \defemph{knows} that~$\exists v$ is true. Indeed, Dolev presented a consensus protocol~$B$ (for {\em ``Beep''}) for the crash failure model in which a process decides on the particular value $v=0$ if and only if it {\em knows} $\exists 0$ \cite{DolevBeep}. It follows from \cite{HalMoWa2001} that there must exist an unbeatable protocol dominating~$B$. Clearly, $B$ decides on~0 as soon as possible. When is the earliest time at which it is possible to decide~1 in a protocol in which decisions on~0 use the rule employed by~$B$? 
Intuitively, a process should decide~1 once it knows that the rule for~0 will never hold for any correct process. Namely, 
let $\nnz$ be the fact that no correct process will {\em ever} know that~$\exists 0$ in the current run. 
Clearly, a process cannot decide~1 before it knows $\nnz$, as this would allow a run violating the \Agreement\ property. 
On the other hand, deciding~1 when $\nnz$ is known {\em is sound}, since no process will ever decide~0, because knowing $\exists 0$ is a precondition for deciding~0. 
To turn this argument into a protocol, we need to present a concrete test for when a process knows $\exists 0$ and when it knows $\nnz$. 
This is facilitated by considering message chains between processes at different times. 

\vspace{-0.15cm}

A {\em process-time node} is a pair $\node{i,m}$ referring to process~$i$ at time~$m$. We say that $\node{j,\ell}$ is {\em seen by} $\node{i,m}$ (in a given run~$r$) if 
there exists a message chain from $j$ at time~$\ell$ to~$i$ at time~$m$. 
It will be convenient to consider $\node{j,\ell}$ as being {\em hidden from} $\node{i,m}$ (in~$r$) if both (a) $i$ does not know that $j$ has failed before time~$\ell$ (it sees no node $\node{j',\ell}$ that did {\em not} see~$\node{j,\ell-1}$, which would prove that~$j$ failed earlier), and (b)~$\node{j,\ell}$ is not seen by $\node{i,m}$. 
It is straightforward to efficiently compute whether $\node{j,\ell}$ is {\em hidden from} $\node{i,m}$
in a run with adversary~$\alpha$ based on the communication graph~$\CG_\alpha$. 
Finally, we say that {\em there is a \defemph{hidden path} with respect to~$\node{i,m}$} in run~$r$ if there exists a sequence of processes $j_0,\ldots, j_{m-1},j_m$ such that $\node{j_\ell,\ell}$ is hidden from $\node{i,m}$, for all $\ell=0,\ldots,m$. For an illustration of hidden paths (indeed, of three disjoint hidden paths) with respect to~$\node{i,2}$, see Figure~\ref{lemma1:second}.
Observe that there does not exist a hidden path with respect to~$\node{i,m}$ precisely if, for some time~$\ell<m$, no node~$\node{j,\ell}$ is hidden from $\node{i,m}$. I.e., if for all processes~$j=1,\ldots,n$, either $\node{j,\ell}$ is seen by $\node{i,m}$, or process~$i$ knows at time~$m$ that~$j$ crashed {\em before} time~$\ell$. 

 A process~$i$ knows~$\exists 0$ (and so can decide~0) at time~$m$ iff it starts with initial value~0, or if some $\node{j,0}$ for a process~$j$ with initial value~0 is seen by $\node{i,m}$. 
For deciding~1, a process knows $\nnz$  exactly if  it knows that \defemph{no active process}  currently knows $\exists 0$. 
Based on this, we show that a process~$i$ knows $\nnz$ at time~$m$ precisely if both (a)~$i$ does not know~$\exists 0$, 
and (b)~no hidden path w.r.t.\ $\node{i,m}$ exists.  
As we show in Section~\ref{sec:PA-con}, this 
protocol can be efficiently implemented without the use of large messages. The resulting protocol is shown in \cref{sec:PA-con} to be unbeatable. It is the first unbeatable protocol for consensus.

For $k$-set consensus the set~$\Vals$ of possible initial values contains at least the $k+1$ values, 
$\{0,\ldots,d\}$, $d \geq k$,  the \Validity\ and \Decision\ properties are as in consensus, and the \Agreement\ property is replaced by 

\vspace{-0.15cm}

\begin{itemize}
\item[]{\bf \defemph{k}-Agreement:}\quad The correct processes that decide in~$r$ decide on at most $k$ distinct values. 
\end{itemize}

\vspace{-0.15cm}

As in the case of consensus, the \Validity\ condition implies that knowing $\exists v$ is also a precondition for deciding~$v$ in this variant of consensus.   Our unbeatable solution to $k$-set consensus is a natural generalization of the one for consensus. 
Define~$v\in\Vals$ to be a \defemph{low} value if $v\in\{0,\ldots, k-1\}$.
At the first instance at which a process~$i$  sees a low value, it decides on the minimal low value it has seen. If $i$ has not seen a low value by time~$m$, it can decide on a value provided that there do not exist $k$ process-disjoint hidden paths with respect to $\node{i,m}$.
Again, this translates into a simple condition regarding the existence of at least $k$ hidden nodes from $\node{i,m}$ at all times $\ell=0,\ldots,m$. 
When this condition holds, $i$ decides on the minimal value that it has seen. 
As discussed above, proving unbeatability of this protocol (\cref{thm:optmink}) is a nontrivial challenge. 

It is often of interest to consider {\em uniform} consensus \cite{CBS-uni,Dutta-uni,H86,KR-uni,Raynal04-uni,WTC-uni}  in which the \Agreement\ property is replaced by

\vspace{-0.15cm}

\begin{itemize}
\item[]{\bf Uniform~Agreement:}\quad The  processes that decide in $r$ must all decide on the same value. 
\end{itemize}

\vspace{-0.15cm}

This forces correct processes and faulty ones to act in a consistent manner.
This requirement makes sense only in a setting where failures are benign, and all processes that decide do so according to the protocol. 
Uniformity may be desirable 
when elements outside the system can observe decisions, as in distributed databases when decisions correspond to commitments to values.
As we shall see, the uniformity constraint strengthens the preconditions for decision, resulting in slower protocols. Therefore, it should be
avoided if possible. We present the first unbeatable protocol for uniform consensus. While it is both conceptually and structurally similar to our unbeatable consensus protocol, the proof of its unbeatability (\cref{thm:u-opt}) is significantly more subtle.

In an asynchronous setting, any non-uniform consensus protocol must also solve uniform consensus. Since the study of $k$-set consensus was initially performed in an asynchronous setting, the common version of $k$-set consensus in the literature is a uniform variant, in which $k$-Agreement is replaced by

\vspace{-0.15cm}

\begin{itemize}
\item[]{\bf Uniform \defemph{k}-Agreement:}\quad The processes that decide in~$r$ decide on at most $k$ distinct values. 
\end{itemize}

\vspace{-0.15cm}

We present a protocol for uniform $k$-set consensus, generalizing our unbeatable uniform consensus protocol and building upon our unbeatable (non-uniform) $k$-set consensus one. This protocol strictly dominates all existing protocols in the literature~\cite{CHLT,GGP,GP09,RRT}, and matches the worst-case bounds for this problem. Whether this protocol is unbeatable remains an open question. 

\vspace{-0.2cm}

\section{Preliminary Definitions}
\label{sec:model}

\vspace{-0.35cm}

Our model of computation is a synchronous, message-passing model with 
benign crash failures.
A system has~\mbox{$n\!\ge\!2$} processes denoted by  
$\Proc=\{1,2,\ldots,n\}$. 
Each pair of processes is connected by a two-way communication link,
and each message is tagged with the identity of the sender.
They share a discrete global clock that starts out at time~$0$ and
advances by increments of one. Communication in the system proceeds in
a sequence of \emph{rounds}, with round~$m+1$ taking place between
time~$m$ and time~$m+1$.
Each process starts in some \emph{initial state} at time~$0$,
usually with an \emph{input value} of some kind.
In every round, each process first sends a set of
messages to other processes, then receives messages sent to it
by other processes during the same round, and then performs some
local computation  based on the
messages it has received.

A faulty process fails by \emph{crashing} in some round~$m\ge 1$. 
It behaves correctly in the first~$m-1$ rounds and 
sends no messages from round~$m+1$ on. 
During its crashing round~$m$, the process may succeed in
sending messages on an arbitrary subset of its links. 
We assume that at most~$\tee \leq n-1$ processes fail in any given execution.

A \defemph{failure pattern} describes how processes fail in an execution.
It is a layered graph~$\FP$ whose vertices are process-time
pairs~$\node{i,m}$ for $i\in\Proc$ and $m\ge 0$. 
Such a vertex denotes process~$i$ and time~$m$. 
An edge has the form $(\node{i,m-1},\node{j,m})$ 
and it denotes the fact that a message sent by~$i$ to~$j$ in round~$m$ would be delivered successfully. 
Let~$\Crash(\tee)$ denote the set of failure patterns in which 
at most~$\tee$ crash failures can occur. 
An \defemph{input vector} describes what input the processes receive in an 
execution. The only inputs we consider are initial values that processes obtain at time~0. 
An input vector is thus a tuple $(v_1,\ldots,v_n)$ where~$v_j$ is the input to process~$j$. 
We think of the input vector and the failure pattern as being determined by an external scheduler, and thus a  pair $\alpha=(\vec{v},\FP)$ is called an {\em adversary}.

A \defemph{protocol} describes what messages a process sends and what decisions it takes, 
as a deterministic function
 of its local state
at the start of a round and the messages received during a round.
We assume that a protocol~$P$ has access to the values of~$n$ and~$\tee$,
typically passed to~$P$ as parameters.

A \defemph{run} is a description of an infinite behaviour of the system.
Given a run~$r$ and a time~$m$, 
$r_i(m)$ denote the \defemph{local state} of process~$i$ at time~$m$ in~$r$
and the \defemph{global state} at time $m$ 
is defined to be $r(m)=\node{r_1(m),r_2(m),\ldots,r_n(m)}$.
A protocol~$P$ and an adversary~$\alpha$ uniquely determine a run, 
and we write $r = P[\alpha]$.

Since we restrict attention to benign failure models and focus on decision times and solvability in this paper, Coan showed that it is sufficient to consider {\em full-information} protocols ({\em fip}'s for short), defined below \cite{Coan}. 
There is a convenient way to consider such protocols in our setting. 
With an adversary $\alpha=(\vec{v},\FP)$ we associate a \defemph{communication graph} $\CG_\alpha$, 
consisting of the graph~$\FP$ extended by labelling the initial nodes $\node{j,0}$ with the initial states $v_j$ according to~$\alpha$. 
With every node $\node{i,m}$ we associate a subgraph  $\Ga(i,m)$ of~$\CG_\alpha$, which we think of as $i$'s {\em view} at $\node{i,m}$.
Intuitively, this graph will represent all nodes $\node{j,\ell}$ from which $\node{i,m}$ has heard, and the initial values it has seen. 
Formally, $\Ga(i,m)$ is defined by induction on~$m$. 
$\Ga(i,0)$ consists of the node $\node{i,0}$, labelled by the initial value~$v_i$. 
Assume that $\Ga(1,m),\ldots,\Ga(n,m)$ have been defined, and let $J\subseteq\Proc$ be the set of processes~$j$ such that $j=i$ or $e_j=(\node{j,m},\node{i,m+1})$ is an edge of~$\FP$. Then $\Ga(i,m+1)$ consists of the node $\node{i,m+1}$, the union of all graphs $\Ga(j,m)$ with $j\in J$, and the edges 
$e_j=(\node{j,m},\node{i,m+1})$ for all $j\in J$. 
We say that $(j,\ell)$ is {\em seen} by $\node{i,m}$ if $(j,\ell)$ is a node of $\Ga(i,m)$. Note that this occurs exactly if $\FP$ allows a (Lamport) message chain starting at $\node{j,\ell}$ and ending at $\node{i,m}$.

A full-information protocol $P$ is one in which at every node $\node{i,m}$ of a run $r=P[\alpha]$ the process~$i$ constructs $\Ga(i,m)$ after receiving its round~$m$ nodes, and sends $\Ga(i,m)$ to all other processes in round~$m+1$. In addition, $P$ specifies what decisions $i$ should take at $\node{i,m}$ based on $\Ga(i,m)$.\footnote{Observe that in benign models fip's  
 do not involve exponentially large states nor exponentially large messages. 
In the crash failure model processes need only send the new edges and nodes that the learn about in every round, rather than the graph $\Ga(i,m)$.}
Full-information protocols thus differ only in the decisions taken at the nodes. 
Let $\dec(i,m)$ be the history of decisions taken by~$i$ up to time~$m$. 
Thus, in a run $r=P[\alpha]$, we define the local state 
$r_i(m) = \langle \dec(i,m),\Ga(i,m)\rangle$. 

\vspace{-0.15cm}

\subsection{Knowledge}

\vspace{-0.35cm}

Our construction of unbeatable protocols will be assisted and guided by a knowledge-based analysis, in the spirit of \cite{FHMV,HM1}. We now define only what is needed for the purposes of this paper. For a comprehensive treatment, the reader is referred to~\cite{FHMV}. 
Runs are dynamic objects, changing from one time point to the next. E.g., at one point process~$i$ may be undecided, while at the next it may decide on a value. Similarly, the set of initial values that~$i$ knows about, or has seen, may change over time. In addition, whether a process knows something at a given point can depend on what is true in other runs in which the process has the same information. 
We will therefore consider the truth of facts at {\em points} $(r,m)$---time~$m$ in run~$r$, with respect to a set or runs~$R$ (which we call a \defemph{system}). 
The systems we will be interested will have the form $R_P=R(P,\gamma)$ where $P$ is a protocol and $\gamma=\gamma(\Vals^n,\Fmodel)$ is the set of all adversaries that assign initial values from~$\Vals$ and failures according to~$\Fmodel$. We will write $(R,r,m)\sat A$ to state that fact~$A$ holds, or is satisfied, at $(r,m)$ in the system~$R$.

The truth of some facts can be defined directly. 
For example, the fact $\exists v$ will hold at $(r,m)$ in~$R$
if some process had initial value~$v$ in~$r$. We say that {\em (satisfaction of)} a fact~$A$ is \defemph{well-defined in~$R$} if 
for every point $(r,m)$ with $r\in R$ we can determine whether or not $(R,r,m)\sat A$. 
Satisfaction of~$\exists v$ is thus well defined. We will write $K_iA$ to denote that \defemph{process~$i$ knows~$A$}, and define: 
\begin{definition}[Knowledge]
\label{def:know}
Suppose that~$A$ is well-defined in~$R$. Then: 

\vspace{-0.2cm}

\begin{tabular}{r l c l}
$(R,r,m)$&\hskip-3mm$\sat K_iA$ & iff & $(R,r',m)\sat A$ ~~\mbox{for all~~} $r'\in R$~~\mbox{such that~~}
$r_i(m)=r'_i(m)$.
\end{tabular}
\end{definition}

\vspace{-0.3cm}

Thus, if $A$ is well-defined in~$R$ then Definition~\ref{def:know} makes $K_iA$ well-defined in~$R$. 
The definition can then be applied recursively, to define the truth of $K_jK_iA$ etc. Moreover, any boolean combination of well-defined facts is also well-defined. Knowledge has been used to study a variety of problems in distributed computing. 
We will make use of the following  fundamental connection between knowledge and action in distributed systems. We say that a fact~$A$ is a \defemph{precondition} for process~$i$ performing action~$\sigma$ in~$R$ if 
$(R,r,m)\sat A$ whenever $i$ performs $\sigma$ at a point $(r,m)$ of~$R$. 
\begin{theorem}[Knowledge of Preconditions, \cite{Mono}]
\label{thm:knowprec}
Assume that $R_P=R(P,\gamma)$ is the set of runs of a deterministic protocol~$P$. 
If $A$ is a precondition for~$i$ performing~$\sigma$ in~$R_P$, then $K_iA$ is a precondition for~$i$ performing~$\sigma$ in~$R_P$.
\end{theorem}

\vspace{-0.15cm}

\section{Unbeatable Consensus}
\label{sec:PA-con}

\vspace{-0.35cm}

We are now ready to apply knowledge to design an unbeatable protocol for consensus. We start with the standard version of consensus defined in the Introduction, and consider the crash failure context $\gammacr=\langle\Vals^n,\Crash(\tee)\rangle$, where $\Vals=\{0,1\}$ --- initial values are binary bits. Every protocol~$P$ in this setting determines a system $R_P=R(P,\gamma)$. 
The \Validity\ property of consensus states that $\exists{v}$ is a precondition for deciding~$v$. 
Theorem~\ref{thm:knowprec} immediately implies:
\begin{lemma}\label{lem:know-exists}
$K_i\exists{v}$ is a precondition for~$i$ deciding on value~$v$ in any protocol satisfying the \Validity\
property. 
\end{lemma} 

\vspace{-0.25cm}

Since we restrict attention to full-information protocols, $(R_P,r,m)\sat K_i\exists v$  exactly if a node $\node{j,0}$ with initial value~$v$ is seen by $\node{i,m}$. 
For if not, then a run~$r'$ of the same protocol exists with $r'_i(m)=r_i(m)$ in which all initial values are~$\boldsymbol{\ne}v$. Notice that this depends only on the adversary $\alpha=(\vec{v},\FP)$: 
If $r=P[\alpha]$ and $r'=Q[\alpha]$ then, for all $i$ and $m$ we have $(R_P,r,m)\sat K_i\exists{v}$ iff  
$(R_Q,r',m)\sat K_i\exists{v}$.

While $K_i\exists{v}$ is a necessary condition for deciding~$v$, if $K_i\exists{0}$ is used as  a sufficient condition for~$\decideZ$ then $K_i\exists{1}$ cannot be sufficient for $\decideO$, since this would violate \Agreement: 
Everyone would decide on their own value at time~0. The following is a consensus protocol in which decisions on~0 are performed as soon as possible:

\noindent
\underline{{\bf Protocol}~$\Pz$}
(for an undecided process~$i$ at time~$m$):\\
\begin{tabular}{ll}
\qquad\qquad{\bf if}  ~$K_i\exists{0} $  & {\bf then} $\decideZ$\\
\qquad\qquad{\bf  elseif}  ~$m=\tee+1$
& {\bf then} $\decideO$
\end{tabular}

$\Pz$ is essentially the early stopping protocol from~\cite{DRS}. We know from \cite{HalMoWa2001} that there exists an unbeatable solution to consensus that dominates~$\Pz$. 
A key step in establishing unbeatability in this case is based on the following lemma
(see Appendix~\ref{sec-proofs-consensus} for proofs):  
\begin{lemma}\label{lem:decide-when-0}
If $Q\dom\Pz$ solves consensus, then every active process~$i$ decides~0 in~$Q$ when $K_i\exists{0}$ first holds. 
\end{lemma}

\vspace{-0.25cm}

By the \Agreement\ property, a precondition for deciding~1 is that no correct process \defemph{ever} decides~0. 
By Lemma~\ref{lem:know-exists}, in consensus protocols that dominate~$\Pz$ processes decide~0 as soon as they know $\exists{0}$. 
It follows that a precondition for deciding~1
is that no correct process will {\em ever} know $\exists{0}$ (denoted by $\nnz$). Indeed, by the Knowledge of Preconditions Theorem~\ref{thm:knowprec},  a process deciding~1 must know this fact. It turns out that this is equivalent to knowing that no active process {\em currently knows} $\exists{0}$. Using this we can show:

\begin{lemma}
\label{lem:know-igno2}
In a full-information protocol in~$\gammacr$,  the following facts are equivalent at time~$m$: 
\vspace{-0.25cm}

\begin{itemize}
\item $K_i(\nnz)$ and 
\item $\neg K_i\exists{0}~~\boldsymbol{\&}~$  there is no hidden path w.r.t.~$\node{i,m}$.
\end{itemize}
\end{lemma}

\vspace{-0.25cm}

As long as there is a hidden path w.r.t.~$\node{i,m}$, process~$i$ considers it possible that some process
currently knows~$\exists{0}$. On such a path is excluded, it knows that it is safe to decide~1. 
This leads to an unbeatable protocol in which decisions on~0 occur as soon as possible, and on~1 as soon as a process knows that~0 will never be decided on: 

\noindent
\underline{{\bf Protocol}~$\OptZ$}
 (for an undecided process~$i$ at time~$m$):\\
\begin{tabular}{ll}
\qquad\qquad  {\bf if}~~$K_i\exists{0}$  &  {\bf then} ~$\decideZ$\\
\qquad\qquad  {\bf elseif}~~no hidden path w.r.t.~$\node{i,m}$ exists
&{\bf then} ~$\decideO$
\end{tabular}

Since we assume for simplicity that communication in our protocols is according to the full-information protocol, we only specify how processes decide.
By \cref{lem:decide-when-0,lem:know-igno2}, we have

\begin{theorem}\label{thm:optz}
$\OptZ$ is an unbeatable consensus protocol in $\gammacr$. 
 \end{theorem}

\vspace{-0.3cm}

It is interesting to compare $\OptZ$ with the protocol $\hmwopt$ that was claimed by \cite{HalMoWa2001} to be unbeatable. Both protocols decide~0 when $\exists{0}$ is known, but they differ in the rule for deciding~1. In $\hmwopt$ a process decides~1 following a round in which it has not discovered a new failure. This condition implies the nonexistence of a hidden path, but is strictly weaker than it. E.g., in a run in which  all initial nodes are seen at $\node{i,2}$ but~$i$ has seen one failure in each of the first two rounds, process $i$ decides in~$\OptZ$ 
but does not decide in~$\hmwopt$.  As a result, we have 
\begin{corollary} 
\label{cor:hmw-notopt}
The protocol $\hmwopt$ presented in~\cite{HalMoWa2001} is \defemph{not} unbeatable. 
\end{corollary}

\vspace{-0.35cm}

Neiger and Bazzi in~\cite{NeigerBazzi} extend the results in \cite{HalMoWa2001}, 
to the 
case of  $\Vals = \{0,\ldots,d\}$ for $d>1$.
We remark that~$\OptZ$ can readily be extended to cover the case in which $\Vals$ contains $\{0,\ldots,d\}$ for $d>1$. The rule for~0 is unchanged, and if no hidden path exists a process can decide on the minimal value it has seen. Thus, a process decides~$v$ when it knows~$\exists v$ and that correct processes will never see a smaller value.
We call this protocol $\OptMin$. In Section \ref{sec-set-consensus}, we show how to extend
$\OptMin$ to the general case of $k$-set consensus.

\vspace{-0.2cm}

\subsection{Majority Consensus}

\vspace{-0.35cm}

Can we obtain other unbeatable consensus protocols? Clearly, the symmetric protocol $\OptO$, obtained from~$\OptZ$ by reversing the roles of~0 and~1, is unbeatable and neither dominates, nor is dominated by, $\OptZ$. 
Of course, $\OptZ$ and~$\OptO$ are extremely biased, each deciding on its favourite value if at all possible, even if only one process has it as an initial value. 
One may argue that it is natural, and may be preferable in many applications, to seek a more balanced solution, in which minority values are not favoured. Fix~$n>0$ and define the fact ``$\Maj=0$'' to be true if more than $n/2$ initial values are~0, while ``$\Maj=1$'' is true if at least 
$n/2$ values are~1. Finally, for a node $\node{i,m}$, we define $\Majvals{i,m}\eqdef 0$ if more than half of the processes whose initial value is known to~$i$ at time~$m$ have initial value~$0$; $\Majvals{i,m}\eqdef 1$ otherwise. 
Consider the following protocol:

\noindent
\underline{{\bf Protocol}~$\OptMaj$}
 (for an undecided process~$i$ at time~$m$):\\
\begin{tabular}{ll}
\qquad\qquad  {\bf if}~~$K_i(\Maj=0)$  &  {\bf then} ~$\decideZ$\\
\qquad\qquad  {\bf elseif}~$K_i(\Maj=1)$ &{\bf then} ~$\decideO$\\
\qquad\qquad  {\bf elseif} no hidden path w.r.t.~$\node{i,m}$ exists & {\bf then} $\decide(\Majvals{i,m})$.
\end{tabular}

\begin{theorem}
\label{thm:optmaj}
If $t>0$, then
$\OptMaj$ is an unbeatable consensus protocol in $\gammacr$. 
\end{theorem}

\vspace{-0.3cm}

Thus, $\OptMaj$ is an unbeatable consensus protocol that satisfies a much stricter validity condition than consensus:
\begin{itemize}
\item[]{\bf Majority Validity:}\quad For $v\in{0,1}$, if more than half of the processes are both correct and have initial value~$v$, then all processes that decide in~$r$ must decide~$v$. 
\end{itemize}

\vspace{-0.2cm}
 
\section{Unbeatable Set Consensus}
\label{sec-set-consensus}

\vspace{-0.35cm}

In this section we present an unbeatable protocol, $\OptMink$, for
(non-uniform) $k$-set consensus.
Recall that for $k$-set consensus the \Agreement\ property of consensus
is replaced with the weaker \kAgreement\ property:
the correct processes decide at most $k$ distinct values.
Since the \Validity\ property of consensus is still required, 
Lemma~\ref{lem:know-exists} applies for $k$-set consensus as well. 

Our protocol $\OptMink$ generalizes the unbeatable consensus protocol $\OptMin$ in Section \ref{sec:PA-con}.
In $\OptMink$, a process $i$ decides on 
a \defemph{low} value (i.e. a value in $\set{0,\ldots,k-1}$) as soon as possible,
namely, the first time $K_i\exists v$ holds, and
decides on a \defemph{high} value $w \in \Vals \setminus \set{0,\ldots,k-1}$, 
as soon as it knows that no $k$ values smaller than $w$ will be decided on.
Recall that $\Vals = \{0,\ldots,d\}$ for some $d \geq k$.

\begin{dfn}
Let $r$ be a run, let~$k$ be a natural number, let $i$ be a process and let $m$ be a time.
We define the following notations, in which
$r$ is implicit.
\vspace{-3mm}
\begin{enumerate}
\item
$\knownvals{i,m}~\eqdef~\{v: K_i \exists v~\mbox{holds at time~$m$}\}$,
\item $\minval{i,m}~\eqdef~\min\knownvals{i,m}$, 
\item $\knownlows{i,m}~\eqdef~ \knownvals{i,m} \cap \set{0,\ldots,k-1}$, and 
\item Process~$i$ is called \defemph{low} at time~$m$ if $\knownlows{i,m} \ne \emptyset$; 
Otherwise, we say that $i$ is \defemph{high} at~$m$.
\end{enumerate}
\end{dfn}

\vspace{-0.25cm}

As already mentioned, in $\OptMink$ low nodes decide immediately. In order
to formalize the decision rule for high nodes, we first formalize the notion
of the amount of process-disjoint hidden paths with respect to a node.

\begin{dfn}
\label{hiddencapacity}
Let $i$ be a process and let $m$ be a time. We define
the \defemph{hidden capacity} of $\node{i,m}$ (in given run) to be the maximum
number $c$ such that  for every time $\ell\le m$, there exist $c$ distinct
processes $i_1^{\ell},\ldots,i_c^{\ell}$ such that  $\node{i_1^{\ell},\ell}$ is hidden from $\node{i,m}$,
for all $\ell\le m$. The nodes $i_b^{\ell}$ are said to be \defemph{witnesses to the
hidden capacity of} $\node{i,m}$.
\end{dfn}

\vspace{-0.35cm}

Analogously to hidden paths, as illustrated in \cref{fig-lemma1} in the Appendix~\ref{sec-proofs-k-set}
a hidden capacity of $c$ indicates that as many as $c$ unknown low values may exist in the system. 
(See Lemma~\ref{exist-hidden-channels} in the Appendix~\ref{sec-proofs-k-set}).
As with hidden paths, it is straightforward to compute whether the hidden capacity of a node $\node{i,m}$
in a run with adversary~$\alpha$ based on the communication graph~$\CG_\alpha$. The hidden capacity
of $\node{i,m}$ can also be very efficiently calculated from the hidden capacity of $\node{i,m\!-\!1}$ using auxiliary
data calculated during the calculation of the latter.
Using these definitions, we phrase a protocol for $k$-set consensus.

\noindent
\underline{{\bf Protocol}~$\OptMink$}
 (for an undecided process~$i$ at time~$m$):\\
\begin{tabular}{ll}
\qquad\qquad {\bf if}~~$i$ is low or~$i$ has hidden capacity $<k$ ~~  {\bf then} ~$\decide(\minval{i,m})$
\end{tabular}

The main technical challenge in proving that $\OptMink$ is unbeatable is,
roughly speaking, showing that e.g.\ in the scenario depicted in \cref{fig-lemma1}, each of the ``hidden'' processes at time $m\!=\!2$ decides on the
unique low value known to it, in \emph{any} protocol that dominates $\OptMink$.
(See \cref{two-face} in Appendix~\ref{sec-proofs-k-set}). We conclude that if~$i$ is high, then it cannot decide without violating \kAgreement. (See \cref{k-set-cant-decide-before} in Appendix~\ref{sec-proofs-k-set}).
We give two proofs for \cref{two-face}. The first is completely constructive, and devoid of any topological arguments
(Appendix~\ref{sec-proofs-k-set}), while the
second is a topological one (Appendix~\ref{sec-topo-proof}). To the best of our knowledge, this is the first result in this field to be given proofs of both kinds,
and a comparative reading sheds light on the relationship between these two dissimilar approaches. 
Our topological proof  reasons in a novel way about subcomplexes of the protocol complex; see \cref{sec-proofs-k-set} for details and a discussion.

The above analysis implies that no $k$-set consensus protocol
can strictly dominate $\OptMink$.
Thus, to prove that $\OptMink$ is unbeatable,
it is enough to show that it indeed solves $k$-set consensus.

\begin{lemma}
\label{k-set-correct}
$\OptMink$ solves $k$-set consensus.
Furthermore, all processes decide in $\OptMink$ by time
$\bigl\lfloor \frac{f}{k} \bigr\rfloor+1$  at the latest.
\end{lemma}

\vspace{-0.2cm}

The proof of Lemma~\ref{k-set-correct} sheds light on an inductive epistemic definition of $\OptMink$, formalizing the intuitive discussion from the beginning
of this section.
Assume that the decision rules for all values $w<v$ have been defined. Define the
decision rule for $v$ as: $i$ decides on $v$ as soon as it knows that (a) $v$
is valid and (b) no more than $k-1$ values~$<v$ will ever be decided upon.

\begin{theorem}
\label{thm:optmink}
$\OptMink$ is an unbeatable $k$-set consensus protocol in $\gammacr$. 
\end{theorem}

\vspace{-0.35cm}

\section{Unbeatable Uniform Consensus}
\label{sec-uniform}

\vspace{-0.35cm}

Under crash failures, a process generally does not know whether or not it is correct. Indeed, so long as it has not seen~$\tee$ other processes crash, the process may (for all it knows) crash in the future. As a result, $K_i\exists{0}$---the rule for deciding~0 in~$\OptZ$---is an inappropriate rule for deciding~0 in any uniform consensus protocol. This is because a process starting with~0 immediately decides~0 with this rule, and may immediately crash. If all other processes have~1, all other decisions can only be on~1. 
Of course, $K_i\exists{0}$ is still a precondition for deciding~0, but it can be strengthened. Denote by $\cv$ the fact ``some \defemph{correct} process knows $\exists{v}$''. We can show the following:
\begin{lemma}
\label{lem:correct-uni}
${K_i\cv}$ is a precondition for~$i$ deciding~$v$ in any protocol solving Uniform Consensus.
\end{lemma}

\vspace{-0.25cm}

\begin{lemma}
\label{lem:u-know}
Let $r\in R_P=R(P,\gammacr)$  and assume that $i$ knows of $\defemph{d}$ failures at $(r,m)$. 
Then \\
$(R_P,r,m)\sat K_i\cv$ ~iff~ one of ~(a)~$m\!>\!0$, $i$ is active at~$m$ and $(R_P,r,m\!-\!1)\sat K_i\exists{v}$, or \\
(b)~~$(R_P,r,m)\sat K_i(\mbox{$K_j\exists{v}$ ~held at time~$m\!-\!1$})$ ~for at least $(\tee\!-\!\defemph{d})$ distinct processes~$j$, holds.
\end{lemma}

\vspace{-0.3cm}

It is easy to check that at time~$\tee+1$ the fact $K_i\exists{v}$ holds exactly if 
at least one of (a) or (b) does;
thus, starting at that time $K_i\exists v$ and $K_i\cv$ are equivalent.
As in the case of consensus, we note that if by time $\tee+1$ we do not have $K_i\exists0$ (equivalently, $K_i\cz$), then we never will.
We thus phrase the following \emph{beatable} algorithm, analogous to $\Pz$, for Uniform Consensus; in this protocol, $K_i\cz$ (the precondition for
deciding~$0$ in uniform consensus) replaces $K_i\exists0$ (the precondition in consensus) as the decision rule for~$0$. The decision rule for~$1$ remains the same.
Note that $K_i\cz$ can be efficiently checked, by \cref{lem:u-know}.

\noindent
\underline{{\bf Protocol}~$\UPz$}
 (for an undecided process $i$ at time $m$):\\
\begin{tabular}{ll}
\qquad\qquad{\bf if}~~$K_i\cz$ & {\bf then}~~$\decideZ$ \\
\qquad\qquad{\bf elseif}~~$m=\tee+1$ & {\bf then}~~$\decideO$.
\end{tabular}

Following a similar line of reasoning that lead us to obtain $\OptZ$, 
we use \cref{lem:know-igno2} to obtain the following unbeatable uniform consensus protocol.

\noindent
\underline{{\bf Protocol}~$\UOptZ$}
 (for an undecided process~$i$ at time~$m$):\\
\begin{tabular}{ll}
\qquad\qquad  {\bf if}~~$K_i\cz$  &  {\bf then} ~$\decideZ$\\
\qquad\qquad  {\bf elseif}~~no hidden path w.r.t.~$\node{i,m}$ exists 
and $\neg K_i\exists{0}$
&{\bf then} ~$\decideO$.
\end{tabular}

\begin{theorem}
\label{thm:u-opt}
$\UOptZ$ is an unbeatable \defemph{uniform} consensus protocol in $\gammacr$. Moreover,
\vspace{-0.3cm}
\begin{itemize}
\item
If $f \ge t-1$, then all decisions are made by time $f+1$ at the latest.
\item
Otherwise, all decisions are made by time $f+2$ at the latest.
\end{itemize}
 \end{theorem}

\vspace{-0.25cm}

Hidden paths again play a central role. Indeed,
as in the construction of $\OptZ$ from $\Pz$, the construction of $\UOptZ$ from $\UPz$ involved some decisions on $1$ being moved forward in time, by means of the last condition, checking the absence of a hidden path. (Decisions on $0$ cannot be moved up, as they are taken as soon as the precondition for deciding $0$ holds.)

Despite the similarity in the design and the structure of the two protocols, the proof of unbeatability for $\UOptZ$ is much more subtle and technically challenging than that for $\OptZ$. This is in a sense since in a uniform consensus protocol dominating $\UOptZ$ (unlike the case of a consensus protocol dominating $\OptZ$),
gaining knowledge even of an initial value of $0$ that is known by a nonfaulty process, no longer implies that some process has already decided on $0$.
As a result, the possibility of dominating $\UOptZ$ by switching~0 decisions to~1 decisions needs to be explicitly rejected. This is done by employing  
reachability arguments essentially establishing the existence of the continual common knowledge conditions of~\cite{HalMoWa2001}
(see the proofs in Appendix~\ref{sec-uni-cons-proofs} for details).

\vspace{-0.3cm}

\subsection{Uniform Set Consensus}
\label{subsec-uni-k}

\vspace{-0.35cm}

We now consider {\em uniform} $k$-set consensus. We present a protocol called $\UOptMink$ that generalizes 
$\UOptZ$ to~$k$ values
(i.e.\ for $k=1$, it behaves exactly
like $\UOptZ$).
While in the protocol $\OptMink$ (which is defined above for non-uniform consensus)
an undecided process~$i$ decides on its minimal value if and only if  
at the time of the decision $i$ is low or has hidden capacity $<k$,
in $\UOptMink$ an undecided process~$i$ decides on a value $v$ if and only if 
$v$ is the minimal value s.t.\ $i$ knows that both a) $v$ was at some stage the minimal value known to a process was low or had hidden capacity $<k$ and b) $v$ will be known to all processes deciding
strictly after $i$.

\noindent
\underline{{\bf Protocol}~$\UOptMink$}
 (for an undecided process~$i$ at time~$m$):\\
\begin{tabular}{ll}
\qquad\qquad{\bf if}~~$\bigl(i$ is low or has hidden capacity $< k\bigr)$ and $K_i\exists\mathsf{correct}(\minval{i,m})$ & {\bf then}~~$\mathsf{decide}_{\minval{i,m}}$ \\
\qquad\qquad{\bf elseif}~~$m>0$ and $\bigl(\node{i,m-1}$ was low or had hidden capacity $< k\bigr)$ & {\bf then}~~$\mathsf{decide}_{\minval{i,m-1}}$ \\
\qquad\qquad{\bf elseif}~~$m=\bigl\lfloor\frac{t}{k}\bigr\rfloor+1$ & {\bf then}~~$\mathsf{decide}_{\minval{i,m}}$
\end{tabular}

As shown by Theorem~\ref{u-k-solve}, 
$\UOptMink$ meets the worst-case lower bounds proven in~\cite{GHP,AGGT} for uniform $k$-set consensus
(see Appendix~\ref{sec-uni-set-cons-proofs} for the proof of Theorem~\ref{u-k-solve}).

\begin{theorem}
\label{u-k-solve}
$\UOptMink$ ~solves \defemph{uniform} $k$-set consensus in $\gammacr$. Moreover,
\vspace{-0.35cm}
\begin{itemize}
\item
If ~$f=t-1\equiv ~0\!\!\mod{k}$, ~
then all decisions are made by time $\frac{f}{k} +1$ at the latest.
\vspace{-0.15cm}
\item
Otherwise, all decisions are made by time $\min\{\bigl\lfloor\frac{t}{k}\bigr\rfloor+1,\lfloor \frac{f}{k} \rfloor +2\}$ at the latest.
\end{itemize}
\end{theorem}

\vspace{-0.35cm}

We emphasize that $\UOptMink$ strictly dominates all existing uniform $k$-set consensus protocols in the literature~\cite{CHLT,GGP,GP09,RRT}.
As in the case of our unbeatable protocols for consensus, uniform consensus, and (nonuniform) $k$-set consensus, the dependence of our protocols on hidden capacity and hidden paths rather than on the number of failures seen often yields much faster stopping times. 
Thus, in runs~$r$ with the property that every correct process discovers exactly~$k$ new failures 
in each of the first $\left\lfloor\frac{f}{k} \right\rfloor$ rounds, all protocols in \cite{CHLT,GGP,GP09,RRT} will decide in more than $\left\lfloor\frac{f}{k} \right\rfloor$ rounds. In contrast, for many of these runs the protocol $\UOptMink$ may be able to decide in as few as~2 rounds. At this point, however, we have been unable to resolve the following

\vspace{-0.15cm}
\noindent{\bf Open Question:}\quad {\it Is ~$\UOptMink$~ an \pdo\ solution to uniform $k$-set consensus in $\gammacr$?}

\vspace{-0.3cm}

\section{Discussion}
\label{sec-discussion}

\vspace{-0.4cm}

Unbeatability is a natural optimality criterion for distributed protocols. 
It formalizes the intuition that a given protocol cannot be strictly improved upon, which is significantly stronger than saying that it is worst-case optimal. When an all-case optimal solution exists, as for simultaneous consensus, an unbeatable protocol will be all-case optimal. 
We have presented the first unbeatable protocols for consensus, uniform consensus and $k$-set consensus. In addition, we suggested a protocol for uniform $k$-set consensus, that strictly dominates all known protocols for the problem.

\vspace{-0.15cm}

Our particular notion of unbeatability, due to Halpern, Moses and Waarts, is based on a natural and commonly accepted notion of domination among protocols~\cite{DRS,GGP,NeigerBazzi,RRT}. Indeed, the original early-stopping protocol $\Pz$ was favoured because it improved on the earlier protocols, and our $\OptZ$ improves upon it. Nevertheless, our notion of unbeatability is just one criterion of this type. Alternative ways to compare runs of different protocols may make sense, depending on the application. One could, for example, compare runs in terms of the time at which the last correct process decides, rather than when each of the processes does. Let us call the corresponding notion \defemph{last-decider unbeatability}.%
\footnote{This notion was suggested to us by Michael Schapira; we thank him for the insight.}  We note that last-decider unbeatability neither implies, nor is implied by, the notion of unbeatability studied so far in this paper. Nevertheless, none of the  protocols previously proposed in the literature for the problems we have studied are last-decider unbeatable. 
In Appendix~\ref{sec-notions} we show that all of our unbeatable protocols are also last-decider unbeatable:
\begin{theorem}
\label{thm:last-decider}
The protocols $\OptZ$, $\OptMaj$, $\OptMink$, and $\UOptZ$ are also last-decider unbeatable for consensus, majority consensus, $k$-set consensus
and uniform consensus, respectively. 
\end{theorem}

\vspace{-0.3cm}

In summary, this paper used a knowledge-based analysis to obtain the first ever {\em unbeatable} protocols for a range of agreement problems in the crash failure model. It identified and exposed hidden paths and hidden capacity as  patterns that play an essential role in determining whether decisions can be taken. As a side effect, we were able to design an unbeatable protocol, $\OptMaj$, for \defemph{majority consensus}, which provides more balanced decision behaviour than previously available in early stopping protocols. 

For ease of exposition and analysis, all of our  protocols were specified under the assumption of full-information message passing. In fact, they can all be implemented in such a way that any process sends any other process a total of $O(n\log n)$ bits throughout the protocol  (see Lemma~\ref{nlogn} in Appendix~\ref{sec-notions}). Thus, unbeatability is attainable at a modest price. Our study opens the way to many possible extensions. For one, we have left open the question of whether $\UOptMink$ is unbeatable. But unbeatability can be sought in other models, and for other problems. Arguably, to be really good, a distributed protocol better be unbeatable!

\bibliographystyle{plain}
\bibliography{z}

\newpage
\appendix

\section{Proofs of Section~\ref{sec:PA-con} --- Consensus}
\label{sec-proofs-consensus}

\cref{lem:know-exists} follows from \cref{thm:knowprec}.

\begin{proof}[Proof of Lemma~\ref{lem:decide-when-0}]
Assume that $Q\dom\Pz$ solves consensus. 
We prove the claim for all processes~$i$ and adversaries~$\alpha$, by induction on the time~$m$ at which 
$K_i\exists 0$ first holds in~$Q[\alpha]$ (and, equivalently, in $\Pz[\alpha]$). 

Base ($m=0$):\quad Since $i$ decides at time~0 in~$\Pz[\alpha]$, it must decide at time~0 in~$Q[\alpha]$ as well. 
At this point we have $K_i\exists{0}$. Since process~$i$ knows only its
initial value at time $0$, it follows that $i$ has initial value $0$. Hence, $K_i\exists{1}$ does \emph{not} hold  at~$0$. By \Validity, $i$ decides~0 in~$Q[\alpha]$.

Inductive step ($m>0$): Assume that the claim holds for all times $<m$.
Recall that $m$ is the first time at which
$K_i\exists{0}$ first holds.
In an fip, this can only happen if $i$ 
receives a message with~0 from some process~$j$ who was active at time~$m-1$.
Thus, $K_j \exists 0$ holds at time $m-1$, and by the induction hypothesis,
$j$ decides $0$ when $K_j\exists{0}$ first holds, which is no later than time $m-1$ in $Q[\alpha]$.
Observe that in~$\gammacr$, if $i$ receives a message from~$j$ in
round~$m$, then~$i$ cannot know that~$j$ is faulty at time~$m$: 
An execution~$\beta$ in which the adversary does not crash~$j$ at all, and that otherwise agrees with $\alpha$ is both legal (initial values are in $\{0,1\}$ and no more than $t$ crash failures) and $Q[\beta]$ is indistinguishable to~$i$ from $Q[\alpha]$ at time~$m$. 
Since $Q$ satisfies \Agreement, $i$ cannot decide~1 at or before 
time~$m$. Moreover, by \Validity, $K_i\exists{0}$ is a precondition for process~$i$ deciding~0, and so $i$ cannot decide $0$ before time $m$. 
Since~$Q$ dominates~$\Pz$, process~$i$ must decide by time~$m$
under~$Q[\alpha]$, and it thus decides~0 at~$m$.
\end{proof}

\begin{proof}[Proof of \cref{lem:know-igno2}]
If $K_i\exists0$, then we immediately have $\lnot K_i(\nnz)$; the fact that the existence of a hidden path implies the possibility for a correct process to know $\exists0$
is generalized by (and implied by) \cref{exist-hidden-channels}, and so its proof is omitted here. The second direction is generalized (and implied) by the proof of the $k$-Agreement property in \cref{k-set-correct},
and so its proof is omitted here as well.
\end{proof}

\cref{thm:optz} follows from \cref{lem:decide-when-0,lem:know-igno2}.

\subsection{Proofs for Majority Consensus}
\label{sec-maj-proofs}
The proof of \cref{thm:optmaj} is based on two lemmas:
\begin{lemma}[Decision at time $1$]\label{maj-decide-first-round}
Assume that  $n\!>\!2$ and $t\!>\!0$. Let $Q\dom\OptMaj$ solve Consensus and let $r\!=\!Q[\alpha]$ be a run of $Q$.
Let $i$ be a process and let $v$ be a value. If $K_i(\Maj\!=\!v)$ at $(r,1)$, then $Q$ makes~$i$ decide~$v$ before or at time $1$ in~$r$.
\end{lemma}

\begin{proof}
By definition of $\OptMaj$, $i$ decides in $\OptMaj[\alpha]$ by time $1$, since $K_i(\Maj\!=\!v)$ holds at $(\OptMaj[\alpha],1)$. As $Q\dom\OptMaj$, we thus have that $i$ must decide upon some value in $r\!=\!Q[\alpha]$ before or at time $1$. Thus,
it is enough to show that $i$ cannot decide $1\!-\!v$ up to time $1$ in $r$.

We prove the claim by induction on $n\!-\!|Z_i|$,
where $Z_i$ is defined to be the set of processes $k$ with initial value $v$, s.t.\ $\node{k,0}$ is seen by $\node{i,1}$.
As $K_i(\Maj\!=\!v)$ at $(r,1)$, we have $|Z_i|\ge \frac{n}{2}$ and so $2 \le |Z_i| \le n$.

Base: $|Z_i|=n$. In this case, all initial values are $v$, and so by \Validity~$i$ cannot decide $1\!-\!v$ in $r$.

Step: Let $2\le\ell<n$ and assume that the claim holds whenever $|Z_i|=\ell+1$. Assume that $|Z_i|=\ell$. As $|Z_i| \ge 2$, there exists $j \in Z_i \setminus \{i\}$. We reason by cases.

\begin{enumerate}[label=\Roman*.]
\item
If there exists a process $k$ s.t.\ $\node{k,0}$ is hidden from $\node{i,1}$, then there exists a run $r'$ of $Q$, s.t.~~\emph{i)} $r'_i(1)\!=\!r_i(1)$,\ \ \emph{ii)}~neither $i$ nor $j$ fail in $r'$,\ \ \emph{iii)} $k$ has initial value $0$ in $r'$, and\ \ \emph{iv)} $Z_j = Z_i\!\cup\!\{k\}$ in $r'$. (Note that by definition, $Z_i$ has the same value in both $r$ and $r'$.)
By the induction hypothesis (switching the roles of $i$ and $j$), $j$ decides $v$ before or at time $1$ at $r'$, and therefore by \Agreement, $i$ cannot decide $1\!-\!v$ in $r'$, and hence it does not decide $1\!-\!v$ up to time $1$ in $r$.
\item
If there exists a process $k \ne i$ with initial value $1\!-\!v$, s.t.\ $\node{k,0}$ is seen by $\node{i,1}$, then $k \notin \{i,j\}$. Hence,
as $t\!>\!0$, there exists a run $r'$ of $Q$, s.t.\ \ \emph{i)} $r'_i(1)\!=\!r_i(1)$,\ \ \emph{ii)}~neither $i$ nor $j$ fail in $r'$,\ \ \emph{iii)} $\node{k,0}$ is hidden from $\node{j,1}$ in $r'$, and\ \ \emph{iv)} $Z_j\!=\!Z_i$ in $r'$. (Once again, $Z_i$ has the same value in both $r$ and $r'$.)
By Case I (switching the roles of $i$ and $j$), $j$ decides $v$ before or at time $1$ in $r'$, and therefore by \Agreement, $i$ cannot decide $1\!-\!v$
in $r'$, and hence it does not decide $1\!-\!v$ up to time $1$ in $r$.
\item
Otherwise, $\node{k,0}$ is seen by $\node{i,1}$ for all processes $k$, and $k$ has initial value $v$ for all processes $k \ne i$. As $|Z_i|<n$, we have that $i$ has
initial value $1\!-\!v$. Thus, there exists a run $r'$ of $Q$, s.t.\ \ \emph{i)} $r'_i(1)\!=\!r_i(1)$,\ \ \emph{ii)}~$f=0$ in $r'$, and\ \ \emph{iii)} $Z_j\!=\!Z_i$ in $r'$. (Once again, $Z_i$ has the same value in both $r$ and $r'$.)
As $i$ has initial value $1\!-\!v$ in $r'$ as well, by Case II (switching the roles of $i$ and $j$), $j$ decides $v$ before or at time $1$ in $r'$, and therefore by \Agreement, $i$ cannot decide $1\!-\!v$
in $r'$, and hence it does not decide $1\!-\!v$ up to time $1$ in $r$, and the proof is complete. \qedhere
\end{enumerate}
\end{proof}

\begin{lemma}[No Earlier Decisions]\label{maj-not-decide}
Assume that $n\!>\!2$ and $t\!>\!0$. Let $Q\dom\OptMaj$ solve Consensus and let $r$ be a run of $Q$.
Let $i$ be a process and let $m$ be a time, s.t.\ $\lnot K_i(\Maj\!=\!0)$ and $\lnot K_i(\Maj\!=\!1)$.
If there exists a hidden path w.r.t.\ $\node{i,m}$, then $i$ does not decide at $(r,m)$.
\end{lemma}

\begin{proof}
Let $v\in\{0,1\}$ be a value. We show that $i$ does not decide $v$ at $(r,m)$.

We first consider the case in which $m\!=\!0$.
In this case, there exists a run $r'$ of $Q$ s.t.\ \ \emph{i)} $r'_i(0)=r_i(0)$,\ \ \emph{ii)} $\Maj\!=\!1\!-\!v$, and\ \ \emph{iii)} $f=0$.
As $f=0$ and $\Maj\!=\!1\!-\!v$ in $r'$, we have $K_i(\Maj\!=\!1\!-\!v)$ at $(r',1)$, and therefore, by \cref{maj-decide-first-round}, $i$ decides $1\!-\!v$ before or at $1$ in $r'$; therefore, $i$ does not decide $v$ at $(r',0)$, and hence neither does it decide $v$ at $(r,0)=(r,m)$.

We turn to the case in which $m\!>\!0$.
As there exists a hidden path w.r.t.\ $\node{i,m}$, for every $0\le\ell\le m$ there exists a process $b_{\ell}$ s.t.\ $\node{b_{\ell},\ell}$ is hidden from $\node{i,m}$.
Thus, there exists a run $r'$ of $Q$ s.t.\ \ \emph{i)} $r'_i(m)=r_i(m)$,\ \ \emph{ii)} $\Maj\!=\!1\!-\!v$,\ \ \emph{iii)} $\node{b_1,1}$ sees $\node{k,0}$ for all processes $k$ (and
therefore $K_{b_1}(\Maj\!=\!1\!-\!v)$ at $1$,\ \ \emph{iv)} $\node{b_{\ell},\ell}$ is seen by $\node{b_{\ell+1},\ell+1}$ for every $1\le\ell<m$, and\ \ \emph{v)} neither $b_m$
nor $i$ fail in $r'$.
We show by induction that $b_{\ell}$ decides $1\!-\!v$ before or at $\ell$ in $r'$, for every $1 \le \ell\le m$.

Base: By \cref{maj-decide-first-round}, $b_1$ decides $1\!-\!v$ before or at $1$ in $r'$.

Step: Let $1<\ell\le m$ and assume that $b_{\ell-1}$ decides $1\!-\!v$ before or at $\ell\!-\!1$ in $r'$. As $\node{b_{\ell-1},\ell\!-\!1}$ is seen by $\node{b_{\ell},\ell}$ in $r'$, there
exists a run $r''\!=\!Q[\gamma]$ of $Q$, s.t.\ \ \emph{i)} $r''_{b_{\ell}}(\ell)\!=\!r'_{b_{\ell}}(\ell)$, and\ \ \emph{ii)} Neither $b_{\ell-1}$ nor $b_{\ell}$ fail in $r''$.
As $\node{b_{\ell-1},\ell\!-\!1}$ is seen by $\node{b_{\ell},\ell}$, and as $r''_{b_{\ell}}(\ell)\!=\!r'_{b_{\ell}}(\ell)$, $b_{\ell-1}$ decides $1\!-\!v$ before or at $\ell\!-\!1$ in $r''$ as well.
As neither $b_{\ell-1}$ nor $b_{\ell}$ fail in $r''$, by \Agreement~$b_{\ell}$ does not decide $v$ before or at $\ell$ in $r''$. As $\node{b_1,1}$ is seen by $\node{b_{\ell},\ell}$ in $r'$,
we have $K_{b_{\ell}}(\Maj\!=\!1\!-\!v)$ at $(r',\ell)$, and therefore also at $(r'',\ell)$. Thus, $b_{\ell}$ decides in $(\OptMaj[\gamma],\ell)$, and therefore $b_{\ell}$ decides
before or at $\ell$ in $r''$, and so it decides $1\!-\!v$ before or at $\ell$ in $r''$, and hence it also decides $1\!-\!v$ before or at $\ell$ in $r'$, and the proof by induction is complete.

As we have shown, $b_m$ decides $1\!-\!v$ in $r'$. As neither $b_m$ nor $i$ fail in $r'$, by \Agreement~$i$ does not decide $v$ at $(r',m)$, and therefore neither
does it decide $v$ at $(r,m)$.
\end{proof}
We can now prove \cref{thm:optmaj}.

\begin{proof}[Proof of \cref{thm:optmaj}]
\Agreement, \Decision\ and \Validity\ are straightforward and left to the reader. If $n\!>\!2$, then unbeatability follows from \cref{maj-not-decide}.
If $n\!=\!1$, then it is straightforward to verify that the single process always decides at time $0$, and so $\OptMaj$ cannot be improved upon.
Finally, if $n\!=\!2$, then it is easy to check that $\OptMaj$ is equivalent to $\OptO$, and so is unbeatable.
\end{proof}

We note that the condition $t\!>\!0$ in \cref{thm:optmaj} cannot be dropped if $n\!>\!2$. Indeed, if $t\!=\!0$ and $n\!>\!2$, then both $\OptZ$ and $\OptO$ (in which some decisions are made at time $0$, and the rest --- at time $1$)
dominate $\OptMaj$ (in which all decisions are made at time $1$).

\section{Proofs of Section~\ref{sec-set-consensus} --- Set Consensus}
\label{sec-proofs-k-set}

\begin{figure}[ht]
    \begin{center}
        \subfigure[$\node{i,2}$ has hidden capacity~$3$.]{
            \label{lemma1:first}
	    \includegraphics[width=0.22\textwidth]{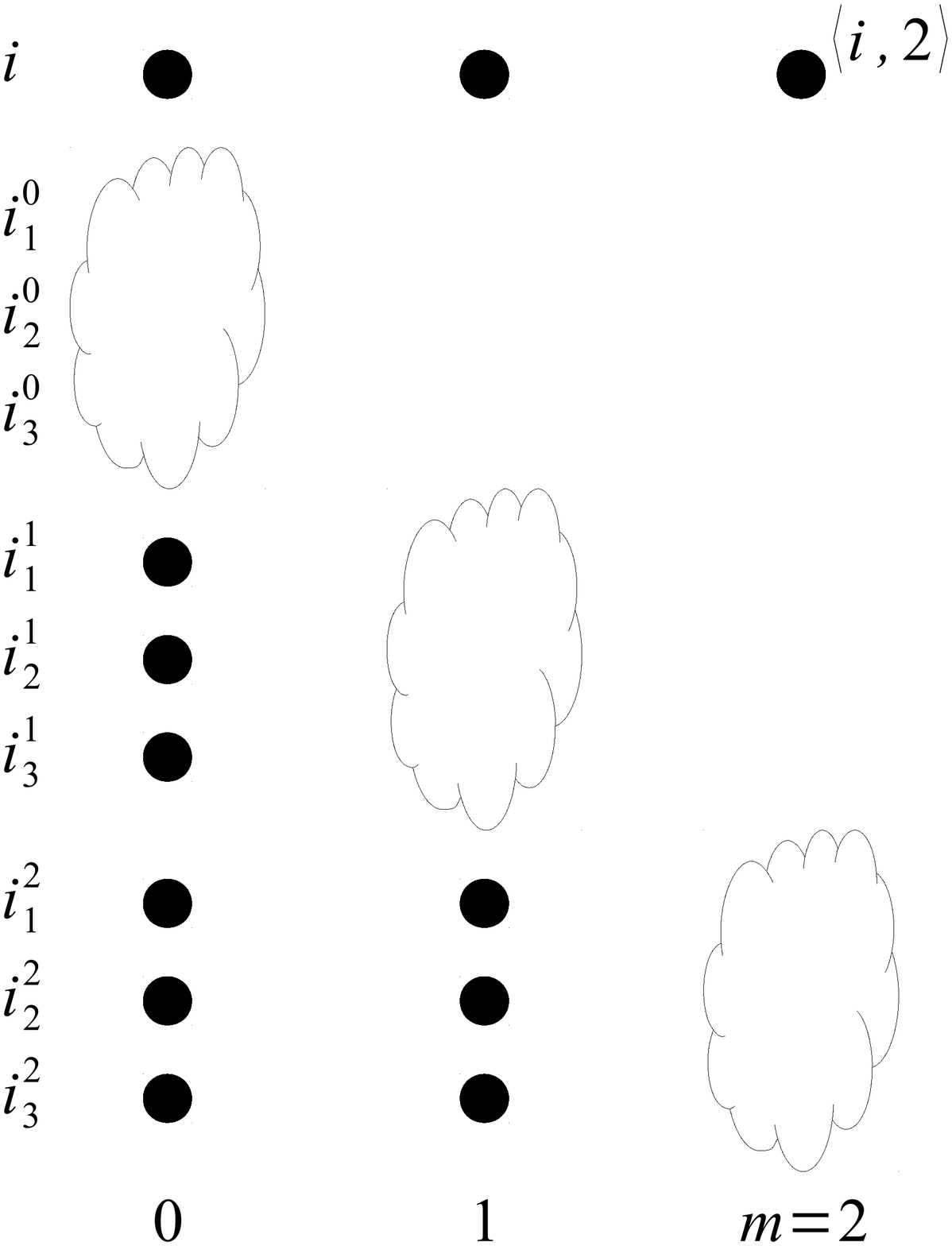}
        }
        \qquad\qquad\qquad\qquad
        \subfigure[A run $i$ considers at $2$ to be possible, in which $v_1,v_2,v_3$ are held by distinct processes at $2$.]{
            \label{lemma1:second}
	    \includegraphics[width=0.22\textwidth]{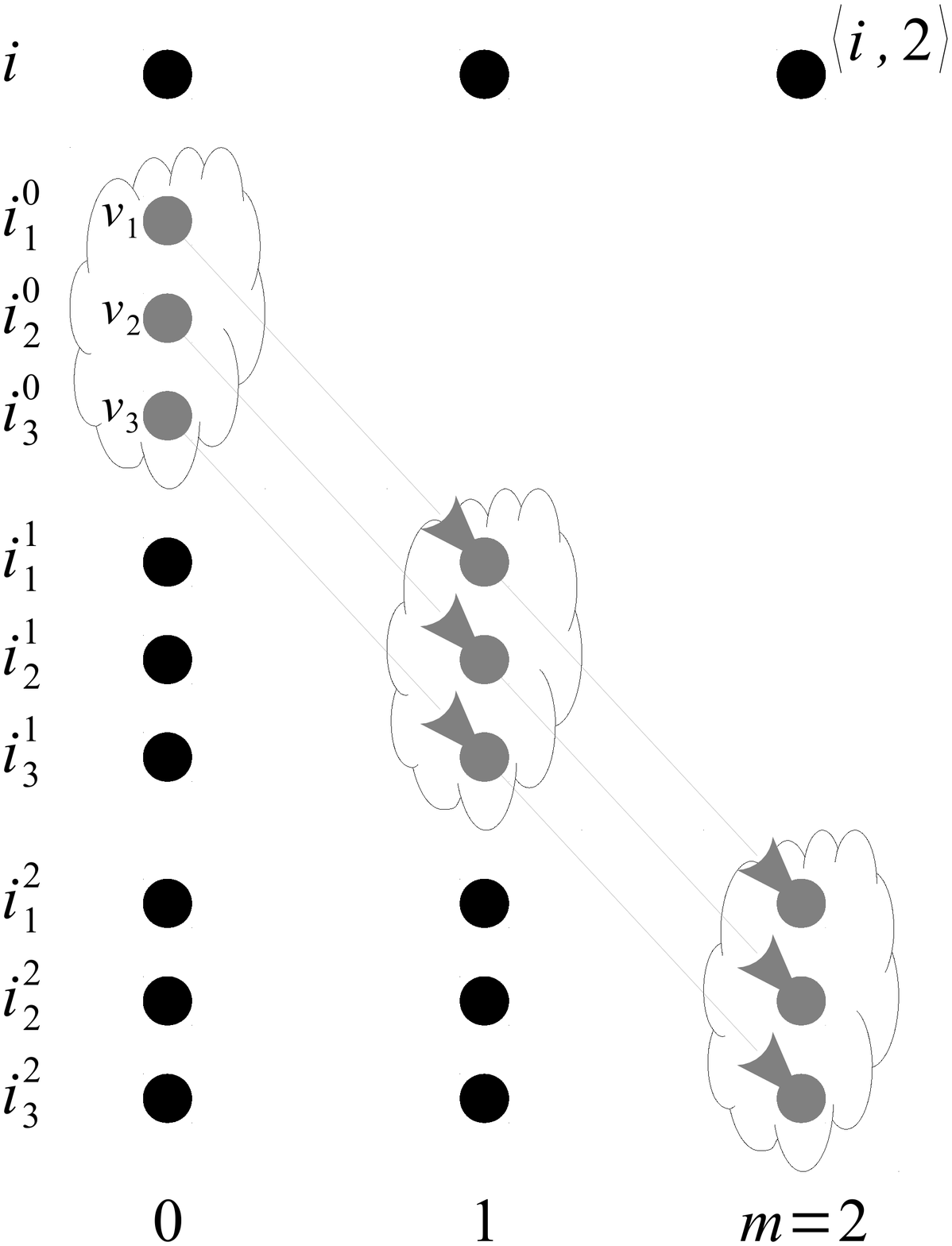}
        }
        \caption{
            A hidden capacity of $c\!=\!3$ at time $m\!=\!2$ indicates that any arbitrary $c$ values unknown to $i$ may exist in the system,
            each held by a distinct process.
        }
	\label{fig-lemma1}
    \end{center}
\end{figure}

\vspace{-4mm}
\begin{remark}
By definition, $\knownvals{i,m}=\emptyset$ (and thus $i$ is high)
for all times $m<0$, for all processes
$i$ in all runs.
\end{remark}

\begin{remark}
The hidden capacity of $i$ in $r$ is (weakly) decreasing as a function of time.
\end{remark}

\begin{lemma}[See Figure~\ref{fig-lemma1}]\label{exist-hidden-channels}
For any fip,
let $r$ be a run, let $i$ be a process and let $m$ be a time such that 
$i$ is active at time $m-1$.
Let $c$ be the hidden capacity of $\node{i,m}$ and
let $i_b^{\ell}$, for all $\ell\le m$ and $b=1,\ldots,c$,
be as in Definition~\ref{hiddencapacity}.
For every $c$ values $v_1,\ldots,v_c$ of $\Vals$, there exists a run $r'$ of the protocol
such that  $r'_i(m)=r_i(m)$,
and for all $\ell$ and $b$, (a)
$v_b \in \knownvals{i_b^{\ell}, \ell}$
(b) $\knownvals{i_b^{\ell},\ell} \setminus \set{v_b} \subseteq \knownvals{i,\ell}$,
and (c) $\node{i_b^{\ell},\ell}$
has hidden capacity $\ge c-1$ witnessed by
$i_{b'}^{\ell'}$ for $b'\ne b$ and $\ell'\le \ell$.
\end{lemma}

\begin{proof}
It is enough to define $r'$ up to the end of round $m$.
Let $i_b^{\ell}$,
for all $\ell\le m$ and $b=1,\ldots,c$,
be as in Definition~\ref{hiddencapacity}.
We define $r'$ to be the same as $r$, except for the following possible changes
(possible, as they may or may not hold in $r$):
\begin{enumerate}
\item
$i_b^0$ is assigned the initial value $b$, for every $b$.
\item
For every $0\le \ell < m$ and every $b$, the process
$i_b^{\ell}$ fails at $\ell$, at which
it successfully sends a message only to $i_b^{\ell+1}$.
\item
For every $0 < \ell \le m$ and every $b$, the process
$i_b^{\ell}$
receives, until time $\ell-1$ inclusive, the
exact same messages as in $r$. (By definition, $\node{i_b^\ell,\ell-1}$ is seen
by $\node{i,m}$ in $r$, and thus it indeed receives messages in $r$ until time
$\ell-1$, inclusive.) At time $\ell$, the process
$i_b^{\ell}$ receives the exact same messages as $i$, and, in addition,
a message from $i$ and the aforementioned message from $i_b^{\ell-1}$.
\end{enumerate}

It is straightforward to check, using backward induction on $\ell$, that in $r'$,
each $\node{i_b^{\ell},\ell}$ is not seen up to time $m$ by any process
other than
$i_b^{\ell'}$ for $\ell'>\ell$,
and is thus hidden from $\node{i,m}$ and from $i_{b'}^{\ell'}$ for all
$b' \ne b$ and for all $\ell'$.
Thus, for all $b$ and $\ell$, $\node{i_b^{\ell},\ell}$
has hidden capacity $\ge c-1$ witnessed by
$i_{b'}^{\ell'}$ for $b'\ne b$ and $\ell'\le \ell$.

We now show that none of the above changes alter the state of $i$ at $m$.
By definition, each $\node{i_b^{\ell},\ell}$ is hidden from $\node{i,m}$
in $r$, and as explained above --- in $r'$ as well.
We note that all modifications above affect a process $i_b^{\ell}$ only at
or after time $\ell$, and as this process at these times is not seen by
$\node{i,m}$
in either run, these modifications do not alter the state of $i$
at $m$.

Let $b \in \set{1,\ldots,c}$.
By definition of $r'$, we have $\knownvals{i_b^0,0} = \set{v_b}$.
Since for every $\ell>0$, $\node{i_b^{\ell},\ell}$ receives a message from
$\node{i_b^{\ell-1},\ell-1}$, we have by induction that
$v_b \in \knownvals{i_b^{\ell},\ell}$ for all $\ell$.

We now complete the proof by showing by induction that for all $\ell$,
$\knownvals{i_b^{\ell},\ell} \subseteq \knownvals{i,\ell} \cup \set{v_b}$.\footnote{%
A similar argument to the one used below in fact further shows that
for all $\ell>0$ and for all $b$,
$\knownvals{i_b^{\ell},\ell} =
\knownvals{i,\ell} \cup \set{v_b}$ in $r'$ for all $\ell$ and $b$.}

Base: $\knownvals{i_b^0,0} = \set{v_b} \subseteq \knownvals{i,0} \cup \set{v_b}$.

Step: Let $\ell>0$.
Let $v \in \knownvals{i_b^{\ell},\ell}$.
If $v \in \knownvals{i_b^{\ell},\ell-1}$, then
$v \in \knownvals{i,\ell}$, as $v_b^{\ell}$
is non-faulty at $\ell-1$ and thus its message is received by $\node{i,\ell}$.
Otherwise, $i_b^{\ell}$ is informed that $\exists v$ by a message it receives
at $\ell$.
By definition of $r'$, a message received by $\node{i_b^{\ell},\ell}$ is exactly
one of the following:
\begin{itemize}
\item
A message received by $\node{i,\ell}$. In this case, $v \in \knownvals{i,\ell}$ as well.
\item
A message sent by $\node{i,\ell-1}$. In this case, we trivially have
$v \in \node{i,\ell-1} \subseteq \knownvals{i,\ell}$.
\item
A message sent by $i_b^{\ell-1}$. In this case, by the induction hypothesis,
\[
v \in \knownvals{i_b^{\ell-1},\ell-1} \subseteq \knownvals{i,\ell-1} \cup \set{v_b}
\subseteq \knownvals{i,\ell} \cup \set{v_b}.\]
\end{itemize}
Thus, the proof by induction, and thus the proof of the lemma, is complete.
\end{proof}

We now generalize
Lemma~\ref{lem:decide-when-0} for $k$-set consensus.
Lemma \ref{two-face} performs this task.

\begin{lemma}
\label{two-face}
Let $P$ be a protocol solving $k$-set consensus.
Assume that in~$P$, every process $i$ that is low at any time~$m$
must decide by time~$m$ at the latest. 
Let $i$ be a process and let~$m$ be a time.
If the following conditions hold in a run $r$:

\begin{enumerate}
\item $i$ does not crash before~$m$,
\item\label{two-face-first-time} $i$ is low at~$m$ for the first time,
\item $\knownlows{i,m}=\set{v}$ for some $v$ (in particular,~$i$ has seen a single low value by time~$m$),
\item\label{two-face-hidden-capacity} $\node{i,m}$ has hidden capacity $\ge k-1$, and 
\item\label{two-face-k-live-ones} there exist $k$ distinct processes $j_1,\ldots,j_k$
such that  $\node{j_b,m-1}$ is high and $\node{j_b,m}$ is hidden from $\node{i,m}$,
for all $b=1,\ldots,k$.
\end{enumerate}
then $i$ decides in~$P$ on its unique low value $v$ at time~$m$. 
\end{lemma}

\begin{remark}
The processes $j_1,\ldots,j_k$ required by
Condition~\ref{two-face-k-live-ones}
of Lemma~\ref{two-face} need not be disjoint from the
processes $i_1^m,\ldots,i_{k-1}^m$ required by
Condition~\ref{two-face-hidden-capacity}.
\end{remark}

\begin{proof}[Proof of Lemma~\ref{two-face}]
We prove the lemma by induction on $m$.

Base ($m=0$):
Since $K_i \exists v$ at time $0$, the value~$v$ must be $i$'s initial value, and thus  
$\knownvals{i,0}=\set{v}$.
As $\node{i,m}$ is low, $i$ decides at $0$.
By the \Validity\ property of $P$, it must decide on a value in $\knownvals{i,0}$,
namely, on $v$.

Step ($m>0$):
\begin{figure}[ht]
    \begin{center}
        \subfigure[$r$, as seen by $\node{i,2}$.]{
            \label{fig-lemma2:first}
	    \includegraphics[width=0.22\textwidth]{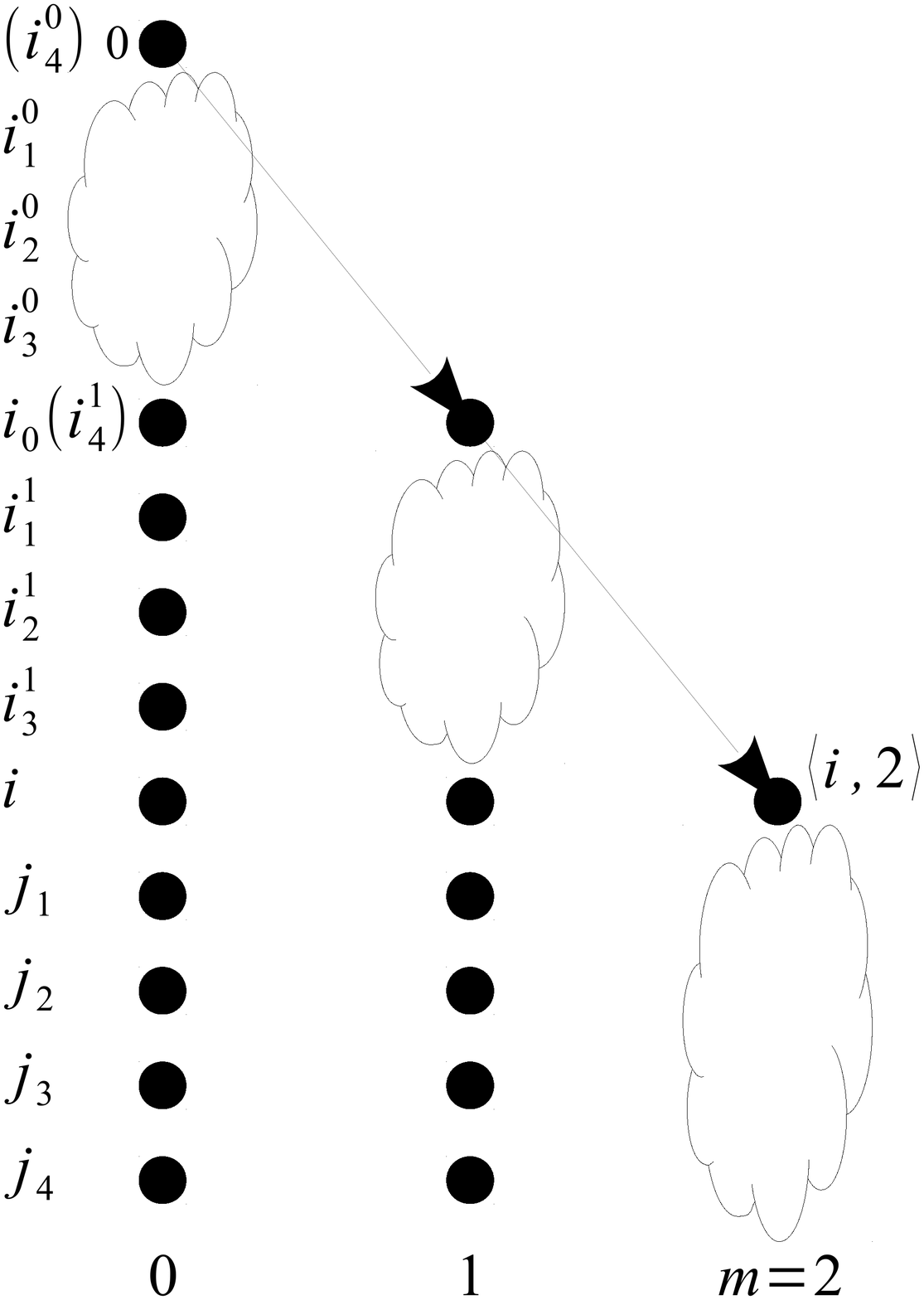}
        }~
        \subfigure[The run $r'$.]{
            \label{fig-lemma2:second}
	    \includegraphics[width=0.22\textwidth]{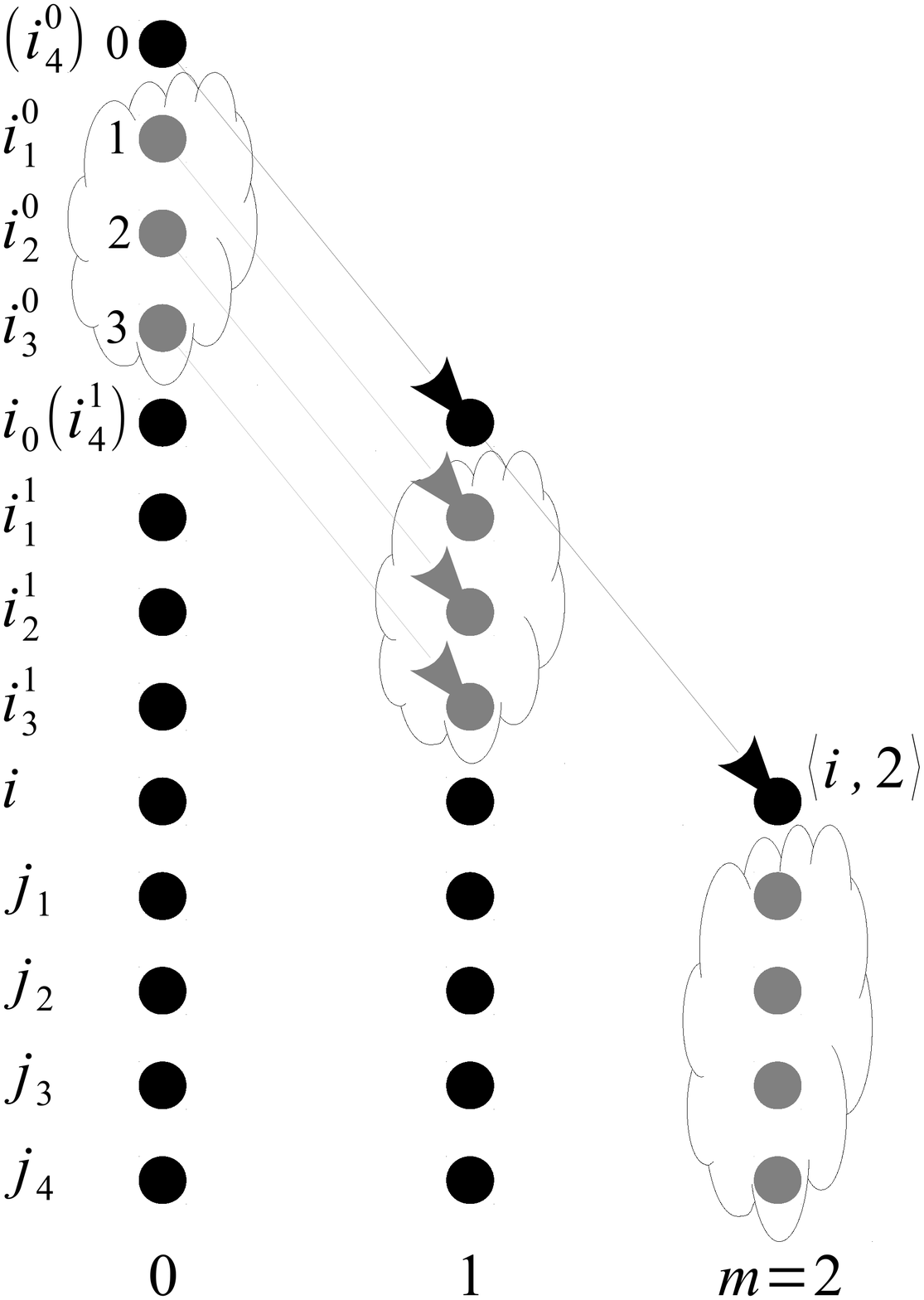}
        }~
        \subfigure[$r'$, as seen by $\node{i_3^1,1}$. The induction hypothesis dictates $i_3^1$ decides $3$ at $1$.]{
            \label{fig-lemma2:third}
	    \includegraphics[width=0.22\textwidth]{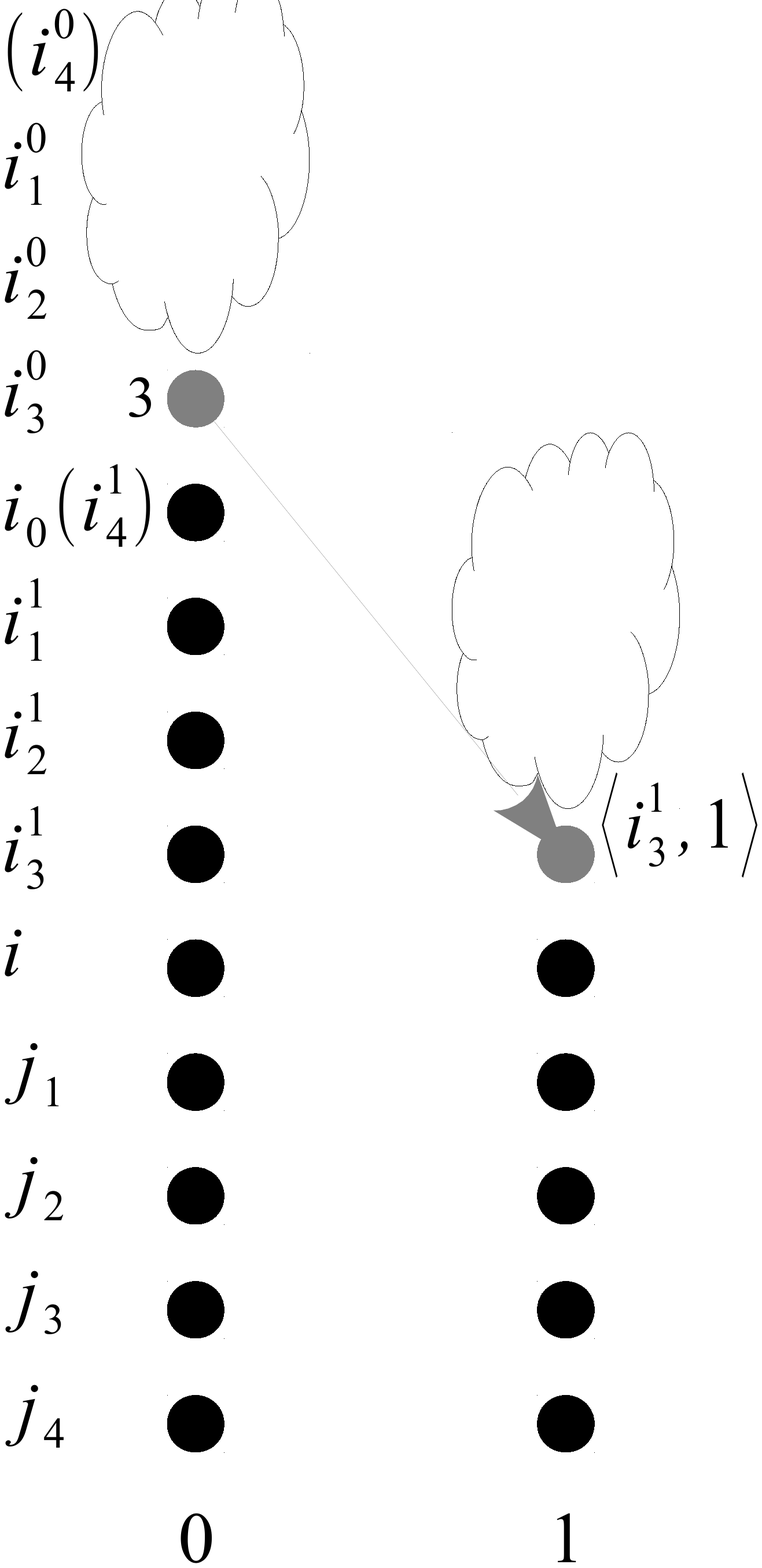}
        }~
        \subfigure[We aim to adjust the messages received by $j_1,\ldots,j_4$ at $2$ so that they collectively decide on all low values.]{
            \label{fig-lemma2:fourth}
	    \includegraphics[width=0.22\textwidth]{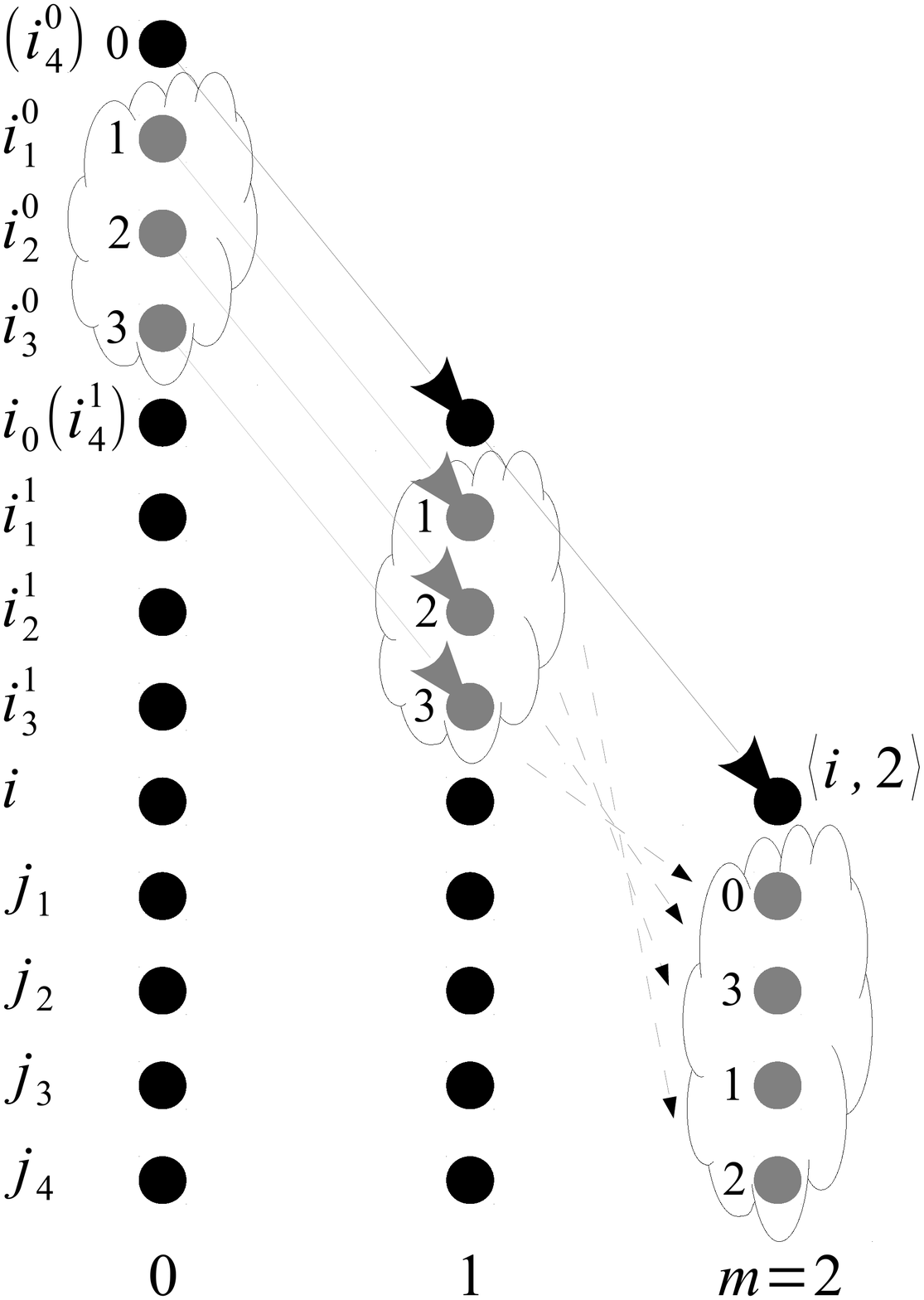}
        }
        \caption{
            Lemma~\ref{two-face} induction step proof strategy (for $m=2$, $k=4$).
        }
	\label{fig-lemma2}
    \end{center}
\end{figure}

Let $i_b^{\ell}$,
for all $\ell\le m$ and $b=1,\ldots,k-1$,
be as in Definition~\ref{hiddencapacity}. (See Figure~\ref{fig-lemma2:first}.)
Let $r'$ be the run of $P$ guaranteed to exist
by Lemma~\ref{exist-hidden-channels},
with respect to the values $\set{0,\ldots,k-1} \setminus \set{v}$. (See
Figure~\ref{fig-lemma2:second}.)
As $j_1,\ldots,j_k$ are seen by $i$ up to time $m$,
we assume w.l.o.g.\ that neither $j_1,\ldots,j_k$ nor $i$ ever fail in $r'$.
We henceforth work in $r'$.

For readability, let us denote by $i_w$, for all
$w \in \set{0,\ldots,k-1} \setminus \set{v}$, the unique process among
the $i_{b'}^{m-1}$ associated with the value $w$ in the definition of $r'$ by
Lemma~\ref{exist-hidden-channels}.
Hence,
$w \in \knownvals{i_w,m-1} \cap \set{0,\ldots,k-1} = \knownlows{i_w,m-1}$.
By Condition~\ref{two-face-first-time}, $\node{i,m-1}$ is high,
and thus, by definition of $i_w$,
$\knownlows{i_w,m-1} = \knownvals{i_w,m-1} \cap \set{0,\ldots,k-1}
\subseteq
(\knownvals{i,m-1} \cap \set{0,\ldots,k-1}) \cup \set{w} =
\knownlows{i,m-1} \cup \set{w} = \set{w}.$
We conclude that $\knownlows{i_w,m-1} = \set{w}$.

As $\knownlows{i,m} \setminus \knownlows{i,m-1} = \set{v} \setminus \emptyset =
\set{v}$, process~$i$ learned that $\exists v$ by a message it received 
at~$m$. Let $i_v$ denote the sender of this message.
We thus trivially have that $v \in \knownlows{i_v,m-1}$.
Furthermore, we have $\knownlows{i_v,m-1} \subseteq \knownlows{i,m} = \set{v}$,
and thus $\knownlows{i_v,m-1} = \set{v}$.

Define $i_k^{m-1}\eqdef i_v$ and $v_k \eqdef v$.
As $\knownlows{i_k^{m-1},m-1} = \set{v}$, for every
$\ell<m-1$ there exists a process $i_k^{\ell}$ s.t.\ (a) $\node{i_k^{\ell},\ell}$ is seen
by $\node{i_k^{\ell+1},\ell+1}$ (and thus does not fail before $\ell$)
and (b) $v \in \knownlows{i_k^{\ell},\ell}$ (and thus
$\knownlows{i_k^{\ell},\ell} = \set{v}$). (See Figure~\ref{fig-lemma2:second}.)
Let $w \in \set{0,\ldots,k-1} \setminus \set{v}$ and let $\ell<m$.
As $\knownlows{i_w,m-1}=\set{w}$,
and as $\knownlows{i_k^{\ell},\ell}=\set{v} \ne \set{w}$,
$\node{i_k^{\ell},\ell}$ is not seen by $\node{i_w,m-1}$
and thus (as $i_k^{\ell}$ does not fail before $\ell$), it is hidden from
$\node{i_w,m-1}$.
Furthermore, as $\knownlows{i_k^{\ell},\ell} = \set{v}$, it is distinct
from all $i_{b}^{\ell}$ for $b<k$.
Let now $w \in \set{0,\ldots,k-1}$.
We conclude that $\node{i_w,m-1}$
has hidden capacity $\ge k-1$ witnessed by
$i_b^{\ell}$ for $\ell \le m-1$ all for all $b$ s.t.\ $v_b \ne w$. (See
Figure~\ref{fig-lemma2:third}.)
Thus, by the induction hypothesis,
$i_w$ decides $w$ by time $m-1$.

We now apply a sequence of consecutive possible
changes to $r'$ (possible, as they may or may not actually modify $r'$),
numbered from $k$ to $1$. (See Figure~\ref{fig-lemma2:fourth}.)
For every $b=1,\ldots,k$,
change~$b$ possibly modifies only~$j_b$, and only at times~$\ge m$,
and does not contradict the fact that $i$ and all $j_1,\ldots,j_k$ never
fail.
Therefore, change~$b$ does not affect the state
$i$ or of $j_{b'}$'s up to time $m$, inclusive. Therefore, once change~$b$
is performed, the state of $j_b$ at $m$ is no longer affected by subsequent
changes.
As we show that following change~$b$, $j_b$ decides at $m$,
and denote the value decided upon by $v_b$, we therefore have that
the fact that $j_b$
decides upon $v_b$ at $m$ at the latest continues to hold throughout the
rest of the changes.

We now inductively describe the changes (recall that changes are performed starting
with change~$k$ and concluding with change~$1$):
Define $r^k\eqdef r'$. For every $b$, change~$b$ is applied to $r^b$
to yield a run $r^{b-1}$.
Let $b \in \set{1,\ldots,k}$ and assume that changes $k,\ldots,b+1$ were already
performed, and that for each $b'>b$, we have that in $r^{b'-1}$ (and thus
in $r^b$), $j_{b'}$ decides
a low value $v_{b'}$ by $m$ at the latest, such that $j_{b+1},\ldots,j_k$ are
distinct of each other.

Change~$b$: Let $j_b$ never fail. Furthermore, let $j_b$ receive at time $m$
messages exactly from (a)
$\set{i_0,\ldots,i_{k-1}} \setminus \set{i_{v_{b+1}},\ldots,i_{v_k}}$,
(b)
$i$, and
(c)
$j_1,\ldots,j_k$, except, of course, from $j_b$.

As $i$ and $j_1,\ldots,j_k$ are all high at $m-1$, and as
$\knownlows{i_w,m-1} = \set{w}$ for all $w$, we now have
$\knownlows{j_b,m} = \set{0,\ldots,k-1} \setminus \set{v_{b+1},\ldots,v_k}$.
In particular, as $b>0$, $\node{j_k,m}$ is low, and therefore must decide
at $m$ or before.
We note that there exists a run $s$ s.t.\ $s_{j_b}(m)=r_{j_b}^{b-1}(m)$,
in which neither $j_b$, nor any of the processes from which it receives messages
at $m$, ever fail. In this run, $j_{b+1},\ldots,j_k$ respectively decide on
$v_{b+1},\ldots,v_k$, and $\set{i_0,\ldots,i_{k-1}} \setminus \set{i_{v_{b+1}},\ldots,i_{v_k}}$ decide on the rest of $\set{0,\ldots,k-1}$. Thus,
by the \kAgreement\ property of $P$,
$j_b$ must decide in $s$ on a value $v_b \in \set{0,\ldots,k-1}$.
As $\knownlows{j_b,m} = \set{0,\ldots,k-1} \setminus \set{v_{b+1},\ldots,v_k}$,
by the \Validity\ property of $P$, we have that $v_b \ne \set{v_{b+1},\ldots,v_k}$.
As $s_{j_b}(m)=r_{j_b}^{b-1}(m)$,
$j_b$ must decide on $v_b$ in $r^{b-1}$ as well and the proof by induction is
complete.

By the above construction, $r_i^0(m)=r'_i(m)=r_i(m)$.
Thus, it is enough to show that in $r^0$, $i$ decides on $v$ at $m$. We thus,
henceforth, work in $r^0$.
As in $r$, and thus also in~$r^0$, $\node{i,m}$ is low, $i$ must decide by
$m$ at the latest.
As all of $j_1,\ldots,j_k$ never fail, and furthermore, collectively
decide on all of $\set{0,\ldots,k-1}$ (see Figure~\ref{fig-lemma2:fourth}), by the \kAgreement\ property of $P$,
as $i$ never fails, it must decide on a low value.
By the \Validity\ property, $i$ must decide on a value known to it to exist.
As $\knownlows{i,m}=\{v\}$ (in~$r$, and thus also in~$r^0$), we have
that $i$ decides $v$. As $v \notin \knownlows{i,m-1}$, by \Validity\
we obtain that $i$ does not decide before $m$ and the proof is complete.
\end{proof}

Using Lemmas~\ref{exist-hidden-channels} and~\ref{two-face},
we derive a necessary condition for deciding in $\OptMink$.

\begin{lemma}
\label{k-set-cant-decide-before}
Let $P$ be a protocol solving $k$-set consensus.
Assume that in~$P$, every process $i$ that is low at any time~$m$ must decide
by time~$m$ at the latest. 
Then no process decides in $P$ as long as it is both high and has hidden capacity
$\ge k$.
\end{lemma}

\begin{proof}
Let $r$ be a run of $P$, let $i$ be a process and let $m$ be a time s.t.\
$\node{i,m}$ is high and has hidden capacity $\ge k$.
Let $i_b^{\ell}$, for all $\ell\le m$ and $b=1,\ldots,k$,
be as in Definition~\ref{hiddencapacity}.
Let $r'$ be the run of $P$ guaranteed to exist
by Lemma~\ref{exist-hidden-channels},
with respect to the values $\set{0,\ldots,k-1}$,
with $i_b^{\ell}$
associated with the value $b-1$ for all $\ell$.
As $i_b^m$, for all $b$, are seen by $i$ up to time $m$,
we assume w.l.o.g.\ that neither they nor $i$ ever fail in $r'$.
As $r'_i(m)=r_i(m)$,
it is enough to show that $i$ does not decide at $m$ in $r'$. We thus,
henceforth, work in $r'$.

Let $b \in \set{0,\ldots,k-1}$. By definition of $r'$,
$\knownlows{i_b^m,m} = \knownvals{i_b^m,m}\cap\set{0,\ldots,k-1} \subseteq
(\knownvals{i,m} \cap \set{0,\ldots,k-1}) \cup \set{b-1} = \knownlows{i,m}
\cup \set{b-1} = \set{b-1}$. As $b-1 \in \knownlows{i_b^m,m}$, we conclude that
$\knownlows{i_b^m,m} = \set{b-1}$.
If $m=0$, then we trivially have that $i_b^m$ is low for the first time at $m$.
Otherwise, as $\node{i,m}$ is high, and as, by definition,
$\node{i_b^m,m-1}$ is seen
by $\node{i,m}$ (in $r$, and therefore in $r'$), we have that $i_b^m$ is
low at $m$ for the first time as well.
By definition of $r'$, $i_b^m$ has hidden capacity $\ge k-1$.
By applying Lemma~\ref{two-face} with $i$ and $\set{i_{b'}^m}_{b' \ne b}$ as
$j_1,\ldots,j_k$, we thus obtain that $i_b^m$ decides $b-1$ at $m$.

Thus, all of $\set{0,\ldots,k-1}$ are decided upon and so, by the \kAgreement\
property of $P$, $i$ may not decide on any other value.
As $\node{i,m}$ is high, by the \Validity\ property of $P$,
$i$ may not decide on any of $\set{0,\ldots,k-1}$ at $m$.
Thus, $i$ does not decide at $m$.
\end{proof}

\begin{proof}[Proof of Lemma~\ref{k-set-correct}]
In some run of $\OptMink$, let $i$ be a non-faulty process.

\Decision:
Let $m$ be a time s.t.\ $i$ has not decided until $m$, inclusive.
Thus, $\node{i,m}$ has hidden capacity $\ge k$.
Let $i_b^{\ell}$, for all $\ell\le m$ and $b=1,\ldots,k$,
be as in Definition~\ref{hiddencapacity}.
By definition, $i_b^{\ell}$, for every $\ell < m$ and $b=1,\ldots,k$, 
fails at time $\ell$. Thus, $k\cdot m \le f$, where $f$ is the number of failure
in the current run. Thus, $m \le \frac{f}{k}$, and therefore
$m \le \bigl\lfloor \frac{f}{k} \bigr\rfloor$.
Therefore, $i$ decides by time $\bigl\lfloor \frac{f}{k} \bigr\rfloor +1$ at the latest.

Henceforth, let $m$ be the decision time of $i$
and let $v = \minval{i,m}$ be the value upon which $i$ decides.

\Validity:
As $v = \minval{i,m}$, we have $v \in \knownvals{i,m}$ and thus
$K_i \exists v$ at $m$. Thus, $\exists v$.

\kAgreement:
It is enough to show that at most $k-1$ distinct values
smaller than $v$ are decided upon in the current run.
Since $i$ decides at $m$, $\node{i,m}$ is either low or has hidden capacity $<k$.
If $\node{i,m}$ is low, then $v = \minval{i,m} \le k-1$, and thus
there do not exist more than $k-1$ distinct legal values smaller than $v$,
let alone ones decided
upon.

For the rest of this proof we assume, therefore, that $\node{i,m}$ is
high and has hidden capacity $<k$.
As $\node{i,m}$ does not have hidden capacity $k$, there exists $0\le\ell\le m$
s.t.\ no more than $k-1$ processes at time $\ell$ are hidden from $\node{i,m}$.

Let $w<v$ be a value decided upon by a non-faulty processor.
Let $j$ be this processor, and let~$m'$ be the time at which $j$ decides on $w$.
As $w<v$ and as $v = \minval{i,m}$, $\node{j,m'}$ is not seen
by $\node{i,m}$.
As $j$ and $i$ are both non-faulty, we conclude that $m' \ge m$, and thus
$m' \ge \ell$.
Let $H$ be the set of all processes seen at $\ell$ by $\node{j,m'}$.
Since $m' \ge \ell$,
We have $\knownvals{j,m'} = \bigcup_{h \in H} \knownvals{h,\ell}$. (Note
that if $m' = \ell$, then $H = \set{j}$.)
As $w = \minval{j,m'}$, we have $w = \minval{h,\ell}$ for
some $h \in H$. As $w < v = \minval{i,m}$, we have
$w \notin \knownvals{i,m}$, and thus $\node{h,\ell}$ is not seen
by $\node{i,m}$. As $\node{h,\ell}$ is seen by $\node{j,m'}$, $h$ has not failed
before $\ell$, and thus $\node{h,\ell}$ is hidden from $\node{i,m}$.
To conclude, we have shown that
\[w \in \bigl\{ \minval{h,\ell} \mid
\mbox{$\node{h,\ell}$ is hidden from $\node{i,m}$} \bigr\}.\]
As there are at most $k-1$ processes hidden at $\ell$ from $\node{i,m}$,
we conclude that no more than $k-1$ distinct values lower than $v$ are
decided upon by non-faulty processes, and the proof is complete.
\end{proof}

\cref{thm:optmink} follows from \cref{k-set-correct,k-set-cant-decide-before}.

\subsection{A Combinatorial Topology Proof of Lemma~\ref{two-face}}
\label{sec-topo-proof}

\subsubsection{Basic Element of Combinatorial Topology}
\label{sec-topo-def}

A \emph{complex} is a finite set $V$
and a collection of subsets $\cK$ of $V$ closed under containment.
An element of $V$ is called a \emph{vertex} of $\cK$,
and a set in $\cK$ is called a \emph{simplex}.
A (proper) subset of a simplex $\sigma$ is called a \emph{(proper) face}.
The \emph{dimension} $\dim \sigma$ is $|\sigma|-1$.
The dimension of a complex $\cK$, $\dim \cK$,
is the maximal dimension of any of $\cK$'s simplexes.
A complex $\cK$ is \emph{pure} if all its simplexes have the same dimension.

For a simplex $\sigma$, let $\bdry \sigma$ denote the
complex containing all proper faces of $\sigma$.  
If $\cK$ and $\cL$ are disjoint, their \emph{join}, $\cK \ast \cL$,
is the complex  $\set{ \sigma \cup \tau : \sigma \in \cK \wedge \tau \in \cL}$.

A \emph{colouring} of a complex $\cK$ is a map from the vertices of
$\cK$ to a set of \emph{colours}.
A simplex of $\cK$ is \emph{fully coloured} if its vertices are mapped to distinct colours.

Informally, a \emph{subdivision} $\Div \sigma$ of $\sigma$ is a complex
constructed by subdividing each $\sigma' \subseteq \sigma$ into smaller simplexes.
A subdivision $\Div \sigma$ maps each $\sigma' \subseteq \sigma$ to the 
pure complex $\Div \sigma'$ of dimension $\dim \sigma$ containing the simplexes that subdivide $\sigma'$.
Thus, for all $\sigma', \sigma'' \subseteq \sigma$,
$\Div \sigma' \cap \Div \sigma'' = \Div \sigma' \cap \sigma''$.
For every vertex $v \in \Div \sigma$,
its \emph{carrier}, $\Car v$,
is the face $\sigma' \subseteq \sigma$ of smallest dimension such that
$v \in \Div \sigma'$.

The \emph{barycentric} subdivision $\Bary \sigma$ of $\sigma$ can be
defined in many equivalent ways. Here we adopt the following 
combinatorial definition.
$\Bary \sigma$ is defined inductively by dimension.
For dimension 0, for every vertex $v$ of $\sigma$,
$\Bary v = v$.
For dimension $\ell$, $1 \leq \ell \leq \dim \sigma$,
for every $\ell$-face $\sigma'$ of $\sigma$,
for a new vertex $v = \sigma'$,
$\Bary \sigma' = v \ast \Bary \bdry \sigma'$.

Let $\Div \sigma$ be a subdivision of $\sigma$. 
A \emph{Sperner colouring} of $\Div \sigma$ is a colouring 
that maps every vertex $v \in \Div \sigma$ to a vertex in $\Car v$.

\subsubsection{Proof of Lemma~\ref{two-face}}

Consider any run $r$. We proceed by induction on the time $m$.

For the base of the induction $m=0$, if the four conditions holds for 
a process $i$ at time $0$, then it must be that 
$i$ starts in $r$ with input $v$, 
and consequently $V\ang{i, m} = \set{v}$.
Therefore, $i$ decides $v$ at time $0$, since $P$ satisfies the validity requirement
of $k$-set consensus.

Let us assume the claim holds until time $m-1$. We prove it holds at $m$.
Let $i$ be a process that satisfies the four conditions at time $m$.

Without loss of generality, let us assume $L\ang{i,m}=\{0\}$.
Let $i_0$ be a process such that $i$ receives a message
from $i_0$ at time $m$ and $L\ang{i_0,m-1} = \set{0}$.
We have $\ang{i,m}$ has hidden capacity greater or equal than $k-1$, 
thus, 
\cref{exist-hidden-channels}
implies that  
there exist a run $r'$ indistinguishable to $\ang{i,m}$ such that 
there are $k-1$ processes $i_1, \hdots, i_{k-1}$ 
such that for each $i_x$, $1 \leq x \leq k-1$,
$\ang{i_x, m-1}$ hidden to $\ang{i,m}$
and $L\ang{i_x,m-1} = \set{x}$.

By induction hypothesis, every $i_x$, $0 \leq x \leq k-1$,
decides at time $m-1$, at the latest, on its unique low value~$x$.
We assume, for the sake of contradiction, that $i$ decides 
on a non-low value at time~$m$
(if $i$ decides before, it necessarily decides on a non-low value).
For simplicity, let us assume $i$ decides on $k$.

By hypothesis, there are $k$ processes, $j_1, \hdots, j_k$, 
(distinct from $i$ and $i_x$)
such for each $1 \leq y \leq k$, $L\ang{j_y, m-1} = \emptyset$. 
Note that $k \in H\ang{j_y, m-1}$, for every $j_y$.

Below, we only consider runs in which 
a subset of $i_0, \hdots, i_{k-1}$ crash in round $m$
and every $j_y$ receives at least one message from some $i_x$;
all other process do not crash in round $m$.
Thus, $L\ang{j_y, m} \neq \emptyset$, for every $j_y$,
and consequently it decides at time $m$, at the latest.

We now define a subdivision, $\Div \sigma$, of a $k$-simplex 
$\sigma = \{ 0, \hdots, k \}$, 
and then define a map $\delta$ from the vertices $\Div \sigma$ to
states of  $i$, $i_x$ or $j_y$ at time~$m$.
The mapping $\delta$ will be defined in a way that
the decisions of the processes induce a Sperner colouring on $\Div \sigma$.
Finally, we argue that,
for every simplex $\tau \in \Div \sigma$, 
all its vertices are mapped to distinct compatible process states
in some execution.
Therefore, by Sperner's Lemma, there must be a $k$-dimensional simplex in $\Div \sigma$
in which $k+1$ distinct values are decided by distinct processes,
thus reaching a contradiction.

\begin{lemma}[Sperner's Lemma]
Let $\Div \sigma$ be a subdivision with a Sperners's colouring $\zeta$.
Then, $\zeta$ defines an odd number of fully coloured $(\dim \sigma)$-simplexes.
\end{lemma}

We construct $\Div \sigma$ inductively by dimension.
The construction is a simple variant of the well-known barycentric subdivision
(see Figure \ref{fig-subdivision} (left)).

For dimension $0$, for every vertex $v \in \sigma$, we define $\Div v = v$;
hence $\Car v = v$.
For every $1$-face (edge) $\sigma'$ of $\sigma$,
if $k \notin \sigma'$ or $\sigma' = \{0, k\}$, then $\Div \sigma' = \sigma'$;
otherwise, for a new vertex $v = \sigma'$, 
$\Div \sigma' = v \ast \Div \bdry \sigma'$.
Note that $\Car v = \sigma'$.
For every $x$-face $\sigma'$ of $\sigma$, $2 \leq x \leq k$,
if $k \notin \sigma'$, then $\Div \sigma' = \sigma'$;
otherwise, for a new vertex $v = \sigma'$, 
$\Div \sigma' = v \ast \Div \bdry \sigma'$.
Again note that $\Car v = \sigma'$
(see Figure \ref{fig-subdivision} (center)).

\begin{figure}[ht]
    \begin{center}
        \includegraphics[scale=.55]{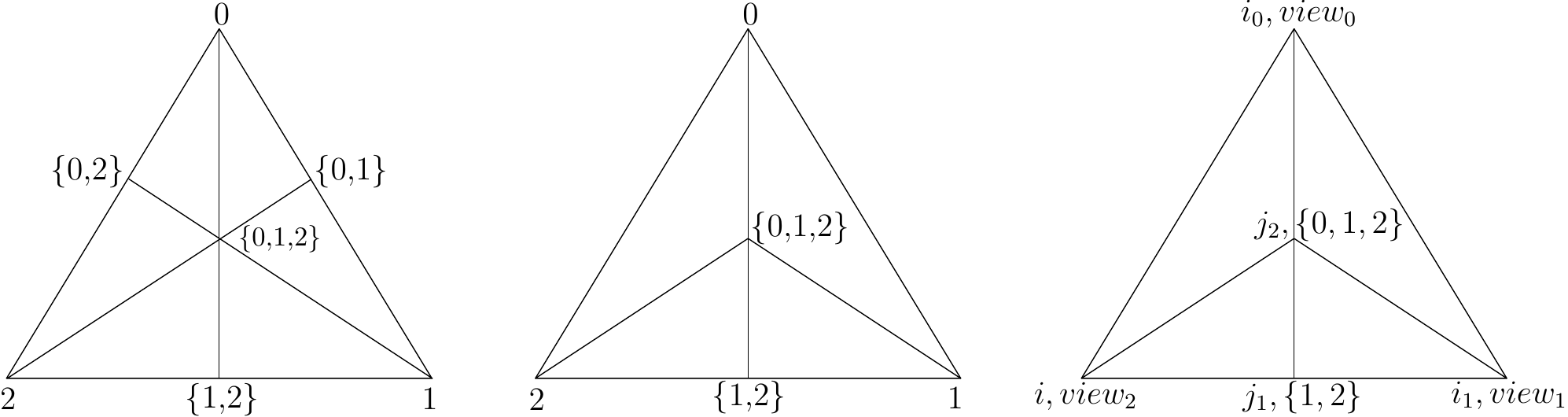}
        \caption{\footnotesize
        For dimension $k=2$ and $\sigma = \set{0,1,2}$,
        the barycentric subdivision $\sigma = \set{0,1,2}$ appears at the left,
		while the subdivision $\Div \sigma$ appears at the center.
        In the subdivision at the right, the vertices are mapped to process states.
        For example, the triangle $\set{\ang{i_0,view_0}, \ang{i_1, view_1}, \ang{j_2, \set{0,1,2}}}$ corresponds
        to the execution in which $i_0$ and $i_1$ do not crash in round $m$, and hence 
        $j_1$ receives $0$ and $1$, which are included in its view.
   		Similarly, the triangle $\set{\ang{i_1, view_1}, \ang{j_1, \set{1,2}}, \ang{j_2, \set{0,1,2}}}$
   		correspond to the execution in which $i_1$ does not crash in round $m$,
   		while $i_0$ crashes and sends a message to $j_2$ and no message to $j_1$.
        The decisions of the processes induces an Sperner colouring: 
        by assumption, $i$ decides $k=2$, and by induction hypothesis, 
        $i_0$ and $i_1$ decided $0$ and $1$ at time $m-1$;
        the rest of the processes have to decide at time $m$ and 
        they can only decide values in their views.
        }
		\label{fig-subdivision}
    \end{center}
\end{figure}

We now define the mapping $\delta$ and the Sperner colouring of $\Div \sigma$, 
which is induced by the decision function $\zeta$ of $P$.

For every vertex $v \in \sigma$, $\Div v = v$.
If $v \neq k$, then $\delta(v)$ is the state $\ang{i_v,m}$ in which 
$i_v$ sends and receives all its messages,
i.e. $i_v$ does not crash in round $m$;
otherwise, $\delta(v) = \ang{i,m}$ in $r'$.
Note that for all $v \in \sigma$, $\zeta(\delta(v)) = v$,
by induction hypothesis and because we assume $i$ decides on $k$. 

For every, $y$-face $\sigma'$ of $\sigma$, $1 \leq y \leq k$,
if there is a vertex $v \in \Div \sigma'$ with $\Car v = \sigma'$,
then $v = \sigma'$ and $k \in \sigma'$.
For such a vertex, we define $\delta(v)$ to be the state $\ang{j_y,m}$ in which 
(a) $j_y$ receives a message from $i_w$, for every $w \in \sigma'$ 
($i_w$ may crash after sending a message to $i_y$), and 
(b) $j_y$ does not receive any message from the $i_{x}$'s whose 
subindexes do not appear in $\sigma'$,
namely, they crash in round $m$
without sending a message to $j_y$.
Observe that $L\ang{i_y,m} = \sigma' \setminus \set{k}$
and $H\ang{i_y,m}$ contains $k$ and possible more high values distinct from $k$.
Since $P$ satisfies the validity requirement
of $k$-set consensus,
$\zeta(\delta(v))$ is any value in $V\ang{i_y,m} = L\ang{i_y,m} \cup H\ang{i_y,m}$. 
For now, we assume that if $\zeta(\delta(v)) \in H\ang{i_y,m}$,
then $\zeta(\delta(v)) = k$, in other words, 
if $i_y$ decides a high value, it decides on $k$;
hence $\zeta(\delta(v)) \in \Car v$.
Therefore, $\zeta$ defines a Sperner colouring for $\Div \sigma$.
Later we explain that this assumption does not affect our argument below.

Consider a $k$-simplex $\tau \in \Div \sigma$.
To show that $\delta$ maps the vertices of $\tau$
distinct process states, it is enough to see that
for every $v \in \Div \sigma$, 
if $\dim \Car v = 0$, then $\delta(v)$ is a state of $i$ or some $i_x$;
and if $1 \leq \dim \Car v \leq k$,
then $\delta(v)$ is a state of $j_y$, where $y = \dim \Car v$.
And to show that $\delta$ map $\tau$ to states of an execution,
note that if there is a $v \in \tau$ such that 
$\delta(v) = \ang{i,m}$, then 
the states in $\delta(\tau)$
correspond to an execution in which 
each $j_y$ receives a subset of the messages from
$i_0, \hdots, i_{k-1}$;
otherwise, the states in $\delta(\tau)$
correspond to an execution in which some 
$i_x$'s distinct from $i_0$ 
do not crash in round $m$ (see Figure \ref{fig-subdivision} (right)).
Observe that in the second case,
the state of $i$ at time $m$ in that execution is different from
the state of $i$ at time $m$ in $r'$, because in $r'$ $i$ only receives
a message from $i_0$.

By Sperner's Lemma, there is at least one fully coloured $k$-simplex in $\Div \sigma$,
and thus there is an execution of $P$ in which $k+1$ distinct values are decided at time $m$.
A contradiction.  

Finally, we assumed that if 
$i_y$ decides a high value, it decides on $k$.
Observe that if in $\Div \sigma$, we replace 
$k$ with the actual decision of $i_y$, then,
the number of distinct decision  at the vertices 
of a simplex of $\Div \sigma$ can only increase.
Thus, in any case, $\Div \sigma$ has a simplex with
$k+1$ distinct decisions. The lemma follows.
$\qed$

\subsection{A Discussion of Unbeatability and Connectivity}
\label{sec-discussion-connectivity}

The topological proof of Lemma \ref{two-face} is more than just a ``trick''
to prove the lemma. Actually, the proof shows what a topological
analysis of unbeatable protocols is about.

In the protocol $\OptMink$, every process decides a low value as soon 
as possible, namely, at the very first time it knows there is low value.
Therefore, if there is a $k$-set consensus protocol $P$ that dominates 
$\OptMink$,
then in $P$ every process $i$ that is low at any time~$m$ 
must decide by time~$m$ at the latest. 

The key in the unbeatability proof for $\OptMink$ is Lemma \ref{two-face},
saying, intuitively, that in every $k$-set consensus protocol $P$ with the property  
that at any time $m$ in which a process $i$ is low for the first time and
has hidden capacity at least $k-1$, then $i$ must decide on a low value.
The topological proof essentially shows that 
$i$ is forced to decide on a low value because the star complex, $\cK$,
of $\ang{i,m}$ in the protocol complex of $P$ at time $m$, $\cP_m$, is $(k-1)$-connected.
Intuitively, $\cK$ is the ``part'' of $\cP_m$ containing all executions
that are indistinguishable to $\ang{i,m}$.
That $\cK$ is $(k-1)$-connected is the reason the proof can map a subdivision of a $k$-simplexes
to process states; indeed, the subdivision is mapped to a subcomplex of $\cK$.
It is well-known that $(k-1)$-connectivity precludes the existence
of a decision function that maps process states to more than $k$ values and avoids a simplex
with $k+1$ distinct decisions at its vertices~\cite{HRT98}, namely, an execution 
with $k+1$ decided values.
Therefore, $i$ has no other choice than decide on a low value,
because if it does not do it, its decision induces an Sperner colouring,
which ultimately implies that the \defemph{k}-\Agreement\ property is violated.

The previous discussion can be formalized in a lemma saying that
if the hypothesis of Lemma~\ref{two-face} hold, then
the star complex $\Star(i,m,\cP_m)$ is $(k-1)$-connected.
Such a lemma can be proved using the techniques in~\cite{HRT98}.

It is worth noticing that in the previous analysis we only 
care about the connectivity of a proper subcomplex of the protocol complex, 
contrary to all known time complexity lower bound proofs~\cite{GHP,HRT98} 
for $k$-set consensus, which care about the 
connectivity of the whole protocol complex.
There is no contradiction with this because these lower bounds proofs
are about the time in which ``all" processes can decide,
which depends on the connectivity of the protocol complex in a given round.
While unbeatability is a notion of optimality concerned with the time at which a ``single"
process can decide, which depends just on a subcomplex 
(the star complex of a given process state) 
of the protocol complex in a given round.

This analysis sheds light on the open question in~\cite{GP09}
about how to extend previous topology techniques to deal with
optimality of protocols. 
In summary, while all-decide lower bounds have to do with
the whole protocol complex, optimal-single-decision lower bounds have to do
with just subcomplexes of the protocol complex. Our topological proof of unbeatability 
here is the first proof that we are aware of that makes this distinction. 

\section{Proofs of Section~\ref{sec-uniform} --- Uniform Consensus}
\label{sec-uni-cons-proofs}

We note that while the assumption $t\!<\!n$ simplifies presentation throughout the proofs below, the case $t\!=\!n$ can be analysed via similar tools.

\begin{proof}[Proof of Lemma~\ref{lem:correct-uni}]
Let~$P$ be a uniform consensus protocol, and let $r$ be a run of~$P$ such that
$(R_P,r,m) \not\sat K_i\cv$.
Thus, there exists a run $r'\in P[\alpha']$ such that $r_i(m)=r'_i(m)$ and
$(R_P,r',m) \not\sat \cv$.
Consider the adversary $\beta$ that agrees with~$\alpha'$ up to time~$m$, and
in which all active but faulty processes at $(r',m)$ crash at time~$m$ without
sending any messages. 
$\beta\in\gammacr$ because it has a legal input vector (identical to~$\alpha'$),
and at most~$\tee$ crash failures, as it has the same set of faulty processes
as~$\alpha'\in\gammacr$. It follows that $r''=P[\beta]$ is a run of~$P$.
Since $\beta$ agrees with $\alpha'$ on the first~$m$ rounds, we have that
$r''_i(m)=r'_i(m)$. Nonetheless, no correct process will ever know $\exists{v}$
in~$r''$, and thus by \Validity\ no correct process ever decides $v$ in~$r''$.
By decision, all correct processes thus decide not on $v$. By \UniAg, and as
$t\!<\!n$ (i.e.\ there are correct processes),
$i$ cannot decide on $v$ in~$r''$, and thus, as $r''_i(m)=r'_i(m)=r_i(m)$,
it cannot decide on $v$ in $r$ at $m$.
\end{proof}

Before moving on to prove \cref{lem:u-know}. We first introduce some notation.

\begin{definition}
For a node $\node{i,m}$, we denote by $\knownf{i,m} \in \{0,\ldots,t\}$ the number of failures known to $\node{i,m}$, i.e.\
the number of processes $j \ne i$ from which $i$ does not receive a message at time $m$.
\end{definition}

We note that \defemph{d}, as defined in \cref{lem:u-know}, is precisely $\knownf{i,m}$.

\begin{proof}[Proof of Lemma~\ref{lem:u-know}]
Sketch:
It is straightforward to see that conditions (a) and (b) imply $K_i\cv$ (Condition (a): as $\node{i,m-1}$ is seen at $m$ by all correct processes; condition (b): as the number of distinct
processes knowing $\exists0$, \emph{including $i$ itself}, is greater than the maximum number of active processes that can yet fail).
If neither condition holds, then $i$ considers it possible that only incorrect processes know $\exists{v}$, and that they all  immediately fail
($i$ at time $m$ before sending any messages, and the others ---
immediately after sending the last message seen by $i$),
in which case no correct process would ever know $\exists{v}$.
\end{proof}

As with~$\Pz$ in the case of consensus, by analysing decisions in protocols dominating $\UPz$, we show that no Uniform Consensus protocol can dominate
$\UOptZ$. \cref{decide-when-0-first-round,decide-when-0} give sufficient conditions for deciding~$0$ in any Uniform Consensus protocol
dominating $\UPz$. As mentioned above, the analysis is considerably subtler 
for Uniform Consensus, because the analogue of \cref{lem:decide-when-0} is not true. Receiving a message with value~0 in a protocol dominating~$\UPz$ does not imply that the sender has decided~0. 

\begin{lemma}[No decision at time $0$]\label{undecided-at-0}
Assume that  $t\!>\!0$.
Let $Q$ solve Uniform Consensus. No process decides at time $0$ in any run of $Q$.
\end{lemma}

\begin{proof}
As $t\!<\!n$, by \cref{lem:correct-uni} it is enough to show that $\lnot K_i \exists v$ for every process $i$ and $v \in \{0,1\}$.
As $0\!<\!t$, and as $\knownf{i,0}=0$ for all processes $i$ by definition, we have that by \cref{lem:u-know}, the proof is complete.
\end{proof}

\begin{lemma}[Decision at time $1$]\label{decide-when-0-first-round}
Let $Q\dom\UPz$ solve Uniform Consensus and let $r=r[\alpha]$ be a run of $Q$. Let $i$ be a process with initial value~$0$ in $r$ s.t. $i$ is  active 
at time $1$ in $r$.
If either of the following hold in $r$, then $\node{i,1}$ decides~$0$ in~$r$.
\begin{parts}
\item\label{decide-when-0-first-round-more-zeros}
$t>0$ ~and~ there exists a process $j \ne i$ with initial value $0$ s.t.\ $\node{j,0}$ is seen by $\node{i,1}$.
\item\label{decide-when-0-first-round-no-more-zeros}
$t>1$ ~and~ $\knownf{i,1} < t$.
\end{parts}
\end{lemma}

\begin{proof}
For both parts, we first note that by \cref{lem:u-know} and by definition of $\UPz$, $i$ decides $0$ at $(\UPz[\alpha],1)$. As $Q\dom\UPz$, we thus have that $i$ must decide upon some
value in $r$ by time $1$.
By \cref{undecided-at-0}, $i$ does not decide at $(r,0)$. Thus, $i$ must decide at $(r,1)$.

We now show \cref{decide-when-0-first-round-more-zeros} by induction on $n\!-\!|Z^0_i|$,
where $Z^0_i$ is defined to be the set of processes $k$ with initial value $0$, s.t.\ $\node{k,0}$ is seen by $\node{i,1}$. Note that by definition, $i,j \in Z^0_i$, and
so $1 < |Z^0_i| \le n$.

Base: $|Z^0_i|=n$. In this case, all initial values are $0$, and so by \Validity~$i$ decides $0$ at $(r,1)$.

Step: Let $1<\ell<n$ and assume that \cref{decide-when-0-first-round-more-zeros} holds whenever $|Z^0_i|=\ell+1$. Assume that $|Z^0_i|=\ell$. We reason by cases.

\begin{enumerate}[label=\Roman*.]
\item
If there exists a process $k$ s.t.\ $\node{k,0}$ is hidden from $\node{i,1}$, then there exists a run $r'$ of $Q$, s.t.~~\emph{i)} $r'_i(1)\!=\!r_i(1)$,\ \ \emph{ii)}~$j$ is active 
at $(r',1)$,~~\emph{iii)} $k$ has initial value $0$ in $r'$, and~~~\emph{iv)} $Z^0_j = Z^0_i\!\cup\!\{k\}$ in $r'$. (Note that by definition, $Z^0_i$ has the same value in both $r$ and $r'$.)
By the induction hypothesis (switching the roles of $i$ and $j$), $j$ decides $0$ at $(r',1)$, and therefore by \UniAg, $i$ cannot decide $1$ at $(r',1)$, and
hence it does not decide $1$ at $(r,1)$. Thus, $i$ decides~$0$ at~$(r,1)$.
\item
Otherwise, $\node{k,0}$ is seen by $\node{i,1}$ for all processes $k$. As $|Z^0_i|<n$, there exists a process $k \notin Z^0_i$ (in particular, $k \notin \{i,j\}$). Hence,
as $t\!>\!0$, there exists a run $r'$ of $Q$, s.t.\ \ \emph{i)} $r'_i(1)\!=\!r_i(1)$,\ \ \emph{ii)}~$j$ is active
at $(r',1)$,\ \ \emph{iii)} $\node{k,0}$ is hidden from $\node{j,1}$ in $r'$, and~~\emph{iv)} $Z^0_j=Z^0_i$ in $r'$. (Once again, $Z^0_i$ has the same value in both $r$ and $r'$.)
By Case I (switching the roles of $i$ and $j$), $j$ decides $0$ at~$(r',1)$, and therefore by \UniAg, $i$ cannot decide $1$
at $(r',1)$, and hence it does not decide $1$ at $(r,1)$. Thus, $i$ decides~$0$ at~$(r,1)$.
\end{enumerate}

We move on to prove \cref{decide-when-0-first-round-no-more-zeros}.
If $\node{k,0}$ is hidden from $\node{i,1}$ for all processes $k \ne i$, then $\lnot K_i\exists1$ at $(r,1)$.
Thus, by \cref{lem:correct-uni}, $i$ cannot decide $1$ at $(r,1)$, and so must decide $0$ at $(r,1)$.
Otherwise, there exists a process $k\ne i$ s.t.\ $\node{k,0}$ is seen by $\node{i,1}$. As $n\!>\!t\!>\!1$, we have $n\!>\!2$ and so
there exists a process
$j \notin \{i,k\}$; if $\knownf{i,1}>0$, then we pick $j$ s.t.\ $\node{j,0}$ is hidden from $\node{i,1}$. Since
$t>1$ (for the case in which  $\knownf{i,1}=0$ and $\node{j,0}$ is seen by $\node{i,1}$) and since $t>\knownf{i,1}$ (for the case in which $\node{j,0}$ is
hidden from $\node{i,1}$), there exists a run $r'$ of $Q$,s.t.\ \ \emph{i)}~$r'_i(1)\!=\!r_i(1)$,\ \ \emph{ii)}~$k$ never fails in $r'$,\ \ \emph{iii)}~$j$ fails at $(r',0)$ before sending any messages except perhaps to $i$, and\ \ \ \emph{iv)}~$i$ fails at $(r',1)$, immediately after deciding but before sending any messages.
Thus, there exists a run $r''$ of $Q$, s.t.\ \ \emph{i)}~$r''_k(m')\!=\!r'_k(m')$ 
\underline{for all} $m'$, 
\ \ \emph{ii)} $k$ never fails in $r''$,\ \ \emph{iii)}~$i$ and~$j$ both have initial value $0$ in $r''$,\ \ \emph{iv)} $j$ fails at $(r'',0)$ while successfully sending a message only to $i$ (and therefore $j \in Z^0_i$ in $r''$), and\ \ \emph{v)} $i$ fails at $(r'',1)$, immediately after deciding but before sending out any messages.
By \cref{decide-when-0-first-round-more-zeros}, $i$ decides $0$ at $(r'',1)$, and therefore $k$ can never decide $1$ during $r''$, and therefore
neither during~$r'$. As $k$ never fails during $r'$, by \Decision\ it must thus decide~$0$ at some point during $r'$. Therefore, by \UniAg, $i$ cannot decide $1$ at $(r',1)$, and thus it does not decide $1$ at $(r,1)$. Thus, $i$ decides $0$ at $(r,1)$.
\end{proof}

\begin{lemma}[Decision at times later than $1$]\label{decide-when-0}
Let $Q\dom\UPz$ solve Uniform Consensus, let $r\!=\!Q[\alpha]$ be a run of $Q$ and let $m\!>\!0$.
Let $i$ be a process s.t.\ $K_i\exists0$ holds at time $m$ for the first time in~$r$, s.t.\ $K_i\cz$ holds at time $m+1$ for the first time in~$r$, and s.t. $i$ is  active 
at $(r,m+1)$.
If either of the following hold in $r$, then $i$ decides~$0$ at~$(r,m+1)$.
\begin{parts}
\item\label{decide-when-0-hidden-z}
All of the following hold.
\begin{itemize}
\item
$\knownf{i,m+1}<t$.
\item
There exists a process $z$ s.t.\ $K_z\exists0$ holds at time $m\!-\!1$, s.t.\ $\node{z,m\!-\!1}$ is seen by $\node{i,m}$, but s.t.\ $\node{z,m}$ is not seen by $\node{i,m\!+\!1}$, 
\item
There exists a process $j \ne i$
s.t.\ $\node{j,m}$ is seen by $\node{i,m\!+\!1}$ and $\node{z,m\!-\!1}$ is seen by $\node{j,m}$.
\end{itemize}
\item\label{decide-when-0-low-knownf}
$\knownf{i,m+1}<t-1$.
\end{parts}
\end{lemma}

\begin{proof}
We prove the lemma by induction on $m$, with the base and the step sharing the same proof (as will be seen below, the conceptual part of an induction base will be played, in a sense, by \cref{decide-when-0-first-round}).

We prove both parts together, highlighting local differences in reasoning for the different parts as needed. For \cref{decide-when-0-low-knownf}, we denote by $z$ an arbitrary process s.t.\ $K_z\exists0$ holds at time $m-1$ and s.t. $\node{z,m\!-\!1}$ is seen by $\node{i,m}$. (As $m>0$, such a process must exist for $i$ to know $\exists0$ at time $m$ for the first time; nonetheless, unlike when proving \cref{decide-when-0-hidden-z},  it is not guaranteed when proving this part that $\node{z,m}$ is not seen by $\node{i,m\!+\!1}$.)

We first note that by \cref{lem:correct-uni} and by definition of $\UPz$, $i$ decides $0$ at $(\UPz[\alpha],m\!+\!1)$. As $Q\dom\UPz$, we thus have that $i$ must decide upon some value in $r$ by time $m\!+\!1$. By \cref{lem:correct-uni}, the precondition for deciding $0$ is not met by $i$ at $(r,m)$. Therefore, it is enough to show
that $i$ does not decide $1$ before or at time $m\!+\!1$ in $r$ in order to show that $i$ decides $0$ at $(r,m\!+\!1)$.

Let $Z^{z,m}_i$ be the set of processes $k$ s.t.\ $\node{k,m}$ is seen by $\node{i,m\!+\!1}$ in $r$ and s.t.\ $\node{z,m-1}$ is seen by $\node{k,m}$ in $r$. (By definition,
$i \in Z^{z,m}_i$.)
Let $C_i$ be the set of all processes $k$ s.t.\ $\node{k,m}$ is either seen by, or hidden from $\node{i,m\!+\!1}$ (i.e.\ the set of nodes that $\node{i,m\!+\!1}$ does not know to be inactive
at time $m$). Note that by definition, $Z^{z,m}_i\subseteq C_i$.
We first consider the case in which $Z^{z,m}_i\supsetneq\{i\}$, and prove the $m$-induction step (for the given $m$) for this case by induction on $|C_i \setminus Z^{z,m}_i|$.

Base: $Z^{z,m}_i=C_i$. In this case, $\node{i,m\!+\!1}$ does not know that $z$ fails 
at time $m\!-\!1$ .
Thus, 
$z \in C_i$ and therefore $z \in Z^{z,m}_i$. 
It follows that $\node{z,m}$ is seen by
$\node{i,m\!+\!1}$ and therefore
the second condition of \cref{decide-when-0-hidden-z} does not hold. Thus, the condition of 
\cref{decide-when-0-low-knownf} holds: 
$\knownf{i,m\!+\!1}  < t\!-\!1$. Furthermore, we thus have that $z$ is active 
at time $m$. We now argue that $z$ decides $0$
at $(r,m)$, which completes the proof of the base case,
as by \UniAg~$i$ can never decide $1$ during $r$. We reason by cases; for both cases, note that since $\node{z,m}$ is seen by $\node{i,m\!+\!1}$, 
we have that 
$\knownf{z,m} \le \knownf{i,m\!+\!1}  < t\!-\!1$.

\begin{itemize}
\item
If $m=1$: As $K_z\exists0$ at time $m\!-\!1=0$, $z$ has initial value $0$.
As $\knownf{z,m}<t\!-\!1$, we have that $t>1$.
By \cref{decide-when-0-first-round-no-more-zeros} of \cref{decide-when-0-first-round} (for $i=z$), we thus have that $z$ decides $0$ at $(r,1)=(r,m)$.
\item
Otherwise, $m\!>\!1$. In this case, as $\node{z,m\!-\!2}$ is seen by $\node{i,m\!-\!1}$, and as $K_i\exists0$ holds at time $m$ for the first time, we have that $K_z\exists0$ holds at time
$m\!-\!1$ for the first time. Similarly, as $\node{z,m\!-\!1}$ is seen by $\node{i,m}$, and as $K_i\cz$ does not hold at time $m$, we have that $K_z\cz$ does not hold at time $m\!-\!1$.
By \cref{decide-when-0-low-knownf} of the $m$-induction hypothesis (for $i=z$), $z$ decides $0$ at $(r,m)$.
\end{itemize}

Step: Let $\{i\}\subsetneq Z^{z,m}_i \subsetneq C_i$, and assume that the claim holds whenever $Z^{z,m}_i$ is of larger size.
For \cref{decide-when-0-hidden-z}, note that $j \in Z^{z,m}_i$, for $j$ as defined in the conditions for that part; for \cref{decide-when-0-low-knownf}, let $j \in Z^{z,m}_i$ be arbitrary.
Analogously to the proof of the induction step in the proof of \cref{decide-when-0-first-round-more-zeros} of \cref{decide-when-0-first-round}, we reason by cases. For the time being, assume that the conditions
of \cref{decide-when-0-low-knownf} hold, i.e.\ that $\knownf{i,m\!+\!1}<t\!-\!1$.

\begin{enumerate}[label=\Roman*.]
\item
If there exists a process $k \in C_i$ s.t.\ $\node{k,m}$ is hidden from $\node{i,m\!+\!1}$, then there exists a run $r'$ of $Q$, s.t.\ \ \emph{i)} $r'_i(m\!+\!1)=r_i(m\!+\!1)$,\ \ \emph{ii)}~$j$ is active 
at $(r',m\!+\!1)$,\ \ \emph{iii)}~$\node{z,m-1}$ is seen by $\node{k,m}$ in $r'$, and\ \ \emph{iv)}~$Z^{z,m}_j = Z^{z,m}_i\!\cup\!\{k\}$ and $C_j=C_i$ in $r'$. (Note that by definition, $Z^{z,m}_i$ and $C_i$ have the same values in both~$r$ and~$r'$.)
We note that $\knownf{j,m\!+\!1}=\knownf{i,m\!+\!1}-1$ in~$r'$, and that by definition $\knownf{i,m\!+\!1}$ is the same in both $r$ and $r'$.
By the 
inductive hypothesis for $Z^{z,m}_j$ (i.e., for $j$ w.r.t.~$z$ at time~$m$), 
$j$ decides $0$ at $(r',m\!+\!1)$, and therefore by \UniAg, $i$ cannot decide $1$ in $r'$,
and therefore it cannot decide $1$ before or at $m\!+\!1$ in $r'$, and the proof is complete.

\item
Otherwise, for each process $k \in C_i$, $\node{k,m}$ is seen by $\node{i,m\!+\!1}$. As $Z^{z,m}_i\subsetneq C_i$,
there exists a process~$k \ne i$ s.t.\ $\node{k,m}$ is seen by $\node{i,m\!+\!1}$ but s.t.\ $\node{z,m\!-\!1}$ is hidden from $\node{k,m}$ (thus $k \ne j$). Hence, and since $\knownf{i,m\!+\!1}<t$, there exists a run $r'$ of $Q$, s.t.\ \ \emph{i)} $r'_i(m\!+\!1)=r_i(m\!+\!1)$,\ \ \emph{ii)}~$j$ is active 
at $(r',m\!+\!1)$,\ \ \emph{iii)}~$\node{k,m}$ is hidden from $\node{j,m\!+\!1}$ in~$r'$, and\ \ \emph{iv)}~$Z^{z,m}_j=Z^{z,m}_i$ and $C_j\supseteq C_i$ in $r'$. (Once again, $Z^{z,m}_i$ and $C_i$ have the same values
in both $r$ and $r'$.) We note that $\knownf{j,m+1}=\knownf{i,m+1}+1$ in $r'$, and that once more, by definition, $\knownf{i,m+1}$ is the same in both $r$ and $r'$.
By Case I (for~$i=j$), and since Case I uses the 
inductive hypothesis for $Z^{z,m}_j$ with one less failure, we conclude that 
$j$ decides~$0$ at $(r',m\!+\!1)$.
Therefor, 
by \UniAg, $i$ cannot decide $1$ at $(r',m\!+\!1)$, and thus it cannot decide~$1$ before or at $m+1$ in $r$, and the proof is complete.
\end{enumerate}

To show that the $Z^{z,m}_i$-induction step also holds under the conditions of \cref{decide-when-0-hidden-z}, we observe that since $\node{z,m}$  is not seen by $\node{i,m\!+\!1}$ in this case, the amount of invocations of Case II
(which uses Case I with one additional known failure) before reaching the $Z^{z,m}_i$-induction base is strictly smaller than that of Case I (which uses the  $Z^{z,m}_i$-induction hypothesis with
one less known failure), and therefore the $Z^{z,m}_i$-induction base is reached with less known failures, i.e.\ with less than $t-1$ known failures, i.e.\ the conditions of \cref{decide-when-0-low-knownf} hold at that point.

Finally, we consider the case in which $Z^{z,m}_i=\{i\}$. As any $j$ as in \cref{decide-when-0-hidden-z} satisfies $j \in Z^{z,m}_i$, we have that the conditions
of \cref{decide-when-0-low-knownf} hold, i.e.\ $\knownf{i,m\!+\!1}<t\!-\!1$. Furthermore, in we have that $\node{z,m}$ is not seen by $\node{i,m\!+\!1}$ (otherwise, $z \in Z^{z,m}_i$). As $\knownf{i,m\!+\!1}<t\!-\!1<n\!-\!2$, there exist two distinct processes $j,k \ne i$ that are not known to $\node{i,m\!+\!1}$ to fail (and thus $i,j,k,z$ are distinct). Thus, $\node{j,m}$ and
$\node{k,m}$ are seen by $\node{i,m\!+\!1}$.

By definition of $j,k$, there exists a run $r'$ of $Q$, s.t.\ \ \emph{i)} $r'_i(m\!+\!1)=r_i(m\!+\!1)$,\ \ \emph{ii)}~$k$ never fails in $r'$,\ \ \emph{iii)}~$j$ fails at $(r',m)$ before sending any messages,\ \ \emph{iv)} $i$ fails at $(r',m+1)$, immediately after deciding but before sending any messages, and\ \ \emph{v)} the faulty processes in~$r'$ are those known by $\node{i,m}$ to fail in $r$, and in addition $i$ and $j$. We note that by definition, $\knownf{i,m\!+\!1}$ is the same in $r$ and $r'$, even though the
number of failures in $r'$ is $\knownf{i,m\!+\!1}+2$.
We notice that there exists a run $r''$ of $Q$, s.t.\ \ \emph{i)}~$r''_k(m')=r''_k(m')$ \underline{for all} $m'$,\ \ \emph{ii)} $k$ never fails in $r''$,\ \ \emph{iii)} $\node{z,m-1}$ is seen by both $\node{i,m}$ and $\node{j,m}$ in $r''$,\ \ \emph{iv)} $j$ fails at $(r'',m)$ while successfully sending a message only to $i$ (and therefore both $j \in Z^{z,m}_i$ and $\knownf{i,m+1}<t-1$ in $r''$), and\ \ \emph{v)} $i$ fails at $(r'',m+1)$, immediately after deciding but before sending out any messages.
By the proof for the case in which $Z^{z,m}_i\supsetneq\{i\}$ ($j\in Z^{z,m}_i$), $i$ decides~$0$ at $(r'',m\!+\!1)$, and therefore $k$ can never decide $0$ during $r''$, and therefore neither during~$r'$. As $k$ never fails during~$r'$, by \Decision\ it must thus decide $0$ at some point during~$r'$. Therefore, by \UniAg, $i$ cannot decide~$1$ before or at $m\!+\!1$ in $r'$, and thus it does not decide~$1$ before or at $m+1$ in $r$, and the proof is complete.
\end{proof}

Now that we have established when processes must decide $0$ in any protocol dominating $\Pz$, we can deduce when processes cannot decide in any such protocol.

\begin{lemma}[No Earlier Decisions when $K_i\exists0$]\label{no-earlier-k0}
Let $Q\dom\UPz$ solve Uniform Consensus, let $r$ be a run of $Q$, let $m$ be a time, and let $i$ be a process.
If at time $m$ in $r$ we have $K_i\exists0$, but $\lnot K_i\cz$, then $i$ does not decide at $(r,m)$.
\end{lemma}

\begin{proof}
If $m\!=\!0$, then by \cref{lem:u-know} and since $\lnot K_i\cz$ at $m\!=\!0$ (even though $K_i\exists0$),
we have $t\!>\!0$. Thus, by \cref{undecided-at-0}, $i$ does not decide at $(r,m)$. Assume henceforth, therefore, that $m\!>\!0$.

As $\lnot K_i\cz$, we have that by \cref{lem:u-know}, $\lnot K_i\exists0$ at time $m\!-\!1$. Thus, there exists a process $z$ s.t.\
$K_z\exists0$ at~$m\!-\!1$, and $\node{z,m\!-\!1}$ is seen by $\node{i,m}$. In turn, by \cref{lem:u-know}, we have that $\knownf{i,m}<t-1$.
There exists a run~$r'$ of $Q$, s.t.\ \ \emph{i)} $r'_i(m)\!=\!r_i(m)$, and\ \ \emph{ii)} the faulty processes in $r'$ are those known by~$\node{i,m}$ to fail in $r$. We henceforth reason about $r'$. By definition of $r'$, $\knownf{i,m\!+\!1}=\knownf{i,m}<t\!-\!1$
(by definition, the value of $\knownf{i,m}$ is the same in both $r$ and $r'$). Thus, by \cref{decide-when-0-low-knownf} of \cref{decide-when-0}, $i$ decides $0$ at $(r',m\!+\!1)$, and hence $i$ does not decide at $(r',m)$,
and therefore neither does it decide at $(r,m)$.
\end{proof}

\begin{lemma}[No Earlier Decisions when $\lnot K_i\exists0$]\label{no-earlier-k1}
Assume that  $t\!>\!0$.
Let $Q\dom\UPz$ solve Uniform Consensus, let $r$ be a run of $Q$, let $m$ be a time, and let $i$ be a process.
If there exists a hidden path w.r.t. $\node{i,m}$ in $r$, and if at time $m$ in $r$ we have $\lnot K_i\exists0$, then $i$ does not
decide at $(r,m)$.
\end{lemma}

\begin{proof}
As $\lnot K_i\exists0$ at time $m$, then by \Validity, $i$ does not decide $0$ at $(r,m)$. Thus, it is enough to show that $i$ does not decide $1$ at $(r,m)$ in order to complete the proof.
If $m\!=\!0$, then by \cref{undecided-at-0}, $i$ does not decide $1$ at $(r,m)$ either. Assume henceforth, therefore, that $m\!>\!0$.

As there exists a hidden path w.r.t. $\node{i,m}$, there exist processes $z,j \ne i$ s.t.\ $\node{z,m\!-\!1}$ is hidden from $\node{i,m}$ and s.t.\
$\node{j,m\!-\!1}$ is seen by $\node{i,m}$.

We first consider the case in which $\knownf{i,m}<t$.
In this case, there exists a run $r'\!=\!Q[\beta]$ of $Q$, s.t.\ all of the following hold in $r'$:
\begin{itemize}
\item
$r'_i(m)=r_i(m)$.
\item
$z$ is the unique process that knows $\exists0$ at $m\!-\!1$, and knows so then for the first time, either having initial value $0$ (if $m\!=\!1$) or (as explained in the Non-Uniform Consensus section) seeing only a single node that knows $\exists0$ at $m\!-\!2$ (if $m\!>\!1)$.
\item
$z$ fails at $(r',m\!-\!1)$, successfully sending messages to all nodes except for $i$.
\item
The faulty processes in $r'$ are those known by $\node{i,m}$ to fail in $r$, and in addition $i$, which fails at time $m$ without
sending out any messages. In particular, $j$ never fails.
\end{itemize}

We henceforth reason about $r'$. First, we note that $\node{j,m\!+\!1}$ does not know that $z$ fails at $m\!-\!1$ (as opposed to at~$m$). As $\node{j,m}$ sees $\node{z,m\!-\!1}$,
as $K_z\exists0$ at $m\!-\!1$, and as $j$ never fails, by \cref{lem:u-know}
we have that $K_j\cz$ at $(r',m\!+\!1)$. Thus, $j$ decides at $(\UPz[\beta],m\!+\!1)$, and so $j$ must decide
before or at~$m\!+\!1$ in $r'$. As $r_i(m)\!=\!r'_i(m)$, then by \UniAg\ it is enough to show that~$j$ does not decide $1$ up to time
$m+1$ in $r'$ in order to complete the proof.

There exists a run $r''$ of $Q$, s.t.\ \ \emph{i)} $r''_j(m\!+\!1)=r'_j(m\!+\!1)$, and\ \ \emph{ii)}~the only difference between $r''$ and $r'$ up to time $m$ is that in $r''$, $z$ fails only at time $m$, after 
deciding 
but without sending a message to $j$. By \UniAg, it is enough to show that $z$ decides $0$ at $(r'',m)$ in order to complete the proof.

We henceforth reason about $r''$. As $z$ does not know at $m$ that neither $z$ nor $i$ fail, we have $\knownf{z,m\!-\!1}\le\knownf{z,m}<t\!-\!1$.
Thus, $t\!>\!1$. If $m\!=\!1$, we therefore have by \cref{decide-when-0-first-round-no-more-zeros} of \cref{decide-when-0-first-round} that $z$ decides $0$ at $(r'',m)$. Otherwise, $m\!>\!1$.
As $K_z\exists0$ at $m\!-\!1$ for the first time, as $\node{z,m\!-\!1}$ sees only one node at $m\!-\!1$ that knows $\exists0$, and as $\knownf{z,m}<t\!-\!1$, by \cref{lem:u-know} we have $\lnot K_z\cz$ at $m\!-\!1$. Thus, by \cref{decide-when-0-low-knownf} of \cref{decide-when-0} (for $i=z$), $z$ decides $0$ at $(r'',m)$. Either way, the proof is complete.

We now consider the case in which $\knownf{i,m}=t$.
There exists a run $r'\!=\!Q[\beta]$ of $Q$, s.t.\ all of the following hold:
\begin{itemize}
\item
$r'_i(m)=r_i(m)$.
\item
All processes $k$ s.t.\ $\node{k,m\!-\!1}$ is hidden from $\node{i,m}$ (including $k=z$) know $\exists0$ at $(r',m\!-\!1)$, either having initial value $0$ (if $m\!=\!1$) or all seeing only a single node that knows $\exists0$ at $m\!-\!2$ (and which fails at time $m\!-\!2$ without being seen by $\node{i,m}$) --- denote this node by $z'$.
\item
All such processes fail at time $m\!-\!1$, successfully sending messages to all nodes except for $i$.
\item
The faulty processes failing in $r'$ are those known by $\node{i,m}$ to fail in $r$. In particular, there are $t$ such processes.
\end{itemize}
We henceforth reason about $r'$.
We note that as $i$ never fails, $\knownf{i,m\!-\!1}\le\knownf{j,m}$ (equality can actually be shown to hold here, but we do not need it).
As the number of nodes at $m\!-\!1$ knowing $\exists0$ that are seen by $\node{j,m}$ equals $\knownf{i,m}-\knownf{i,m\!-\!1}\ge t-\knownf{j,m}$ (by the above
remark, equality holds here as well), we have by \cref{lem:u-know} that $K_j\cz$ at $m$, and therefore $j$ decides
at $(\UPz[\beta],m)$; thus, it must decide before or at $m$ in $r'$. As $r_i(m)\!=\!r'_i(m)$, by \UniAg\ it is enough
to show that $j$ does not decide $1$ up to time $m$ in $r'$ in order to complete the proof.

We proceed with an argument similar in a sense to those of \cref{decide-when-0-first-round-more-zeros} of \cref{decide-when-0-first-round} and the inner induction in the proof of \cref{decide-when-0}.

As $\node{z,m\!-\!1}$ is seen by $\node{j,m}$, there exists a run $r''$ of $Q$, s.t.\ \ \emph{i)} $r''_j(m)\!=\!r'_j(m)$, and\ \ \emph{ii)} the only difference between $r''$ and $r'$ up to time $m$
is that in $r'$, $z$ never fails, but rather $i$ fails at $m\!-\!1$ after sending a message to $j$ but without sending a message to $z$.
We note that there are $t$ processes failing throughout $r''$. 
We henceforth reason about $r''$. If $m\!=\!1$, then $z$ has initial value $0$ and if $m\!>\!1$, then $\node{z,m\!-\!1}$ sees $\node{z',m\!-\!2}$;
either way, by \cref{lem:u-know}, $K_z\cz$ at $(r'',m)$ and therefore $z$ must decide before or at time $m$. Thus, it is enough to show that $z$ does not decide $1$ up to time $m$ in $r''$
in order to complete the proof.

As $\node{i,m\!-\!1}$ is not seen by $\node{z,m}$, there exists a run $r'''$ of $Q$, s.t.\ \ \emph{i)} $r'''_z(m)\!=\!r''_z(m)$, and\ \ \emph{ii)} the only difference between $r'''$ and $r''$ up to time $m$
is that in $r'''$, $\node{i,m-1}$ sees $\node{z',m\!-\!2}$ (or, if $m=1$, then the difference is that $i$ has initial value $0$); we note that $\node{i,m\!-\!1}$ is still seen by $\node{j,m}$.
We note that there are $t$ processes failing throughout $r'''$.
Observe that the number of nodes at $m\!-\!1$ knowing $\exists0$ that are seen by $\node{j,m}$ in $r'''$  is greater than in $r'/r''$ (between which $j$ at $m$ cannot distinguish), however $\knownf{j,m}$ remains the same between $r'/r''$ and~$r'''$; thus, $K_j\cz$ at $m$ in $r'''$ as well, and therefore $j$ must decide before or at time $m$ in $r'''$. Thus, it is enough to show that $j$
does not decide $1$ up to time $m$ in $r'''$ in order to complete the proof. We henceforth reason about~$r'''$. 

As $\node{i,m\!-\!1}$ is seen by $\node{j,m}$, there exists a run $r''''$ of $Q$, s.t.\ \ \emph{i)} $r''''_j(m)=r'''_j(m)$, and\ \ \emph{ii)} the only difference between $r''''$ and $r'''$ up to time $m$
is that in $r''''$, $i$ does not fail (and is thus seen by $\node{z,m}$).
We note that there are~$t-1$ processes failing throughout $r''''$, and thus in particular $\knownf{z,m}<t$. If $m=1$, then by \cref{decide-when-0-first-round-more-zeros} of \cref{decide-when-0-first-round} (for~$i=z$ and $j=i$), $z$ decides $0$ in $(r'''',m)$. Otherwise, i.e.\ if $m\!>\!1$, by \cref{decide-when-0-hidden-z} of \cref{decide-when-0} (for $i=z$, $z=z'$, and~$j=i$), $z$ decides $0$ in $(r'''',m)$. Either way, the proof is complete.
\end{proof}

From \cref{no-earlier-k0,no-earlier-k1}, we deduce sufficient conditions for Unbeatability of Uniform Consensus protocols dominating $\UPz$; these conditions also become necessary
if it can be shown that there exists some Uniform Consensus protocol dominating $\UPz$ that meets them, as we indeed show momentarily for $\UOptZ$.

\begin{cor}\label{cor:uni-opt}
Assume that  $0<t<n$.
A protocol $Q\dom\UPz$ that solves Uniform Consensus and in which a node $\node{i,m}$ decides whenever any of the following hold at $m$, is a \pdo\ Uniform Consensus protocol.
\begin{itemize}
\item
$K_i\cz$.
\item
No hidden path w.r.t.\ $\node{i,m}$ exists, and $\lnot K_i\exists0$.
\end{itemize}
\end{cor}

By \cref{cor:uni-opt}, we have that
if $\UOptZ$ solves Uniform Consensus, then it does so in a \pdo\ fashion.

\begin{lemma}
$\UOptZ\dom\UPz$
\end{lemma}

\begin{proof}
As explained above, at time $\tee+1$ no hidden paths exist, and furthermore, $K_i\exists0$ iff $K_i\cz$.
\end{proof}

\begin{theorem}
\label{u-solve}
$\UOptZ$ ~solves Uniform Consensus in $\gammacr$. Furthermore,
\begin{itemize}
\item
If $f \ge t-1$, then all decisions are made by time $f+1$ at the latest.
\item
Otherwise, all decisions are made by time $f+2$ at the latest.
\end{itemize}
 \end{theorem}

\begin{proof}
This is a special case of \cref{u-k-solve}, for which a complete proof is given below.
\end{proof}

\cref{thm:u-opt} follows from \cref{cor:uni-opt,u-solve}; in the boundary case of $t\!=\!0$ (which is not covered by \cref{cor:uni-opt}), we note that $\UOptZ$ and $\OptZ$ coincide, as do the problems of uniform consensus
and consensus; hence $\UOptZ$ is unbeatable, and \cref{thm:u-opt} holds, in that case as well.

\section{Proofs of Section~\ref{subsec-uni-k} --- Uniform Set Consensus}
\label{sec-uni-set-cons-proofs}

\begin{proof}[Proof of \cref{u-k-solve}]

\Decision:
By definition of $\UOptMink$, every process that is active at time $\bigl\lfloor\frac{t}{k}\bigr\rfloor+1$,
and in particular every non-faulty process, decides by this time at the latest.

Before moving on to show \Validity\ and \UnikAg, we first complete the analysis of stopping times. 
In some run of $\UOptMink$, let $i$ be a process and let $m$ be a time s.t.\ $i$ is active at $m$ but has not decided until $m$, inclusive.
Let $\tilde{m}\le m$ be the latest time not later than $m$ s.t.\ $\node{i,\tilde{m}}$ has hidden capacity $\ge k$. By definition
of $\UOptMink$, as $i$ is undecided at $m$, we have $\tilde{m} \ge m-1$.

As $\node{i,\tilde{m}}$ has hidden capacity $\ge k$ at $\tilde{m}$, let $i_b^{\ell}$, for all $0\le\ell\le \tilde{m}$ and $b=1,\ldots,k$, be as in Definition~\ref{hiddencapacity}.
By definition, $\node{i_b^{\ell},\ell}$, for every $0\le\ell < \tilde{m}$ and $b=1,\ldots,k$, is hidden from $\node{i,\tilde{m}}$.
Thus, $k\cdot \tilde{m} \le \knownf{i,\tilde{m}} \le f$.
therefore, $\tilde{m} \le \frac{f}{k}$ and so $\tilde{m} \le \bigl\lfloor\frac{f}{k}\bigr\rfloor$.
Hence, as $m-1 \le \tilde{m}$, we have $m \le \tilde{m}+1 \le \bigl\lfloor \frac{f}{k} \bigr\rfloor + 1$.
We thus have that every process that is active at time $\bigl\lfloor\frac{f}{k}\bigr\rfloor+2$, decides by this time at the latest.

Assume now that $m = \bigl\lfloor \frac{f}{k} \bigr\rfloor + 1$ and that $f$ is a multiple of $k$.
($i$ is still a process that is active but undecided at $m$.)
As $f$ is a multiple of $k$, then $m = \frac{f}{k} +1$, and so $f = k \cdot (m-1) $. As $f = k \cdot (m-1) \le k \cdot \tilde{m} \le \knownf{i,\tilde{m}} \le  \knownf{i,m}\le f$, we
we have that both $\tilde{m}=m-1$ and $\knownf{i,m}=f$.
As $\tilde{m}=m-1$, we have that $i$ has hidden capacity $<k$ at~$m > \tilde{m}$. As $i$ is undecided at $m$,
we thus have, by definition of $\UOptMink$, that $\lnot K_i\cv$ for $v\eqdef\minval{i,m}$.
As by definition $K_i\exists v$ at $m$, we have by \cref{lem:u-know} that $K_i\exists v$ at $m$ for the first time.
Therefore, as $m>\tilde{m}\ge0$, there exists a process $j$ such that $K_j\exists v$ at $m-1$ and s.t.\ $\node{j,m-1}$ is seen by $\node{i,m}$. Thus, by \cref{lem:u-know}
and since $\lnot K_o\cv$, we have $\knownf{i,m}<t-1$, and so $f=\knownf{i,m}<t-1$.

We thus have that if $f=t-1$ and if this value is a multiple of $k$, then every process that is active at time $\bigl\lfloor\frac{f}{k}\bigr\rfloor+1$ decides by this time at the latest.

We move on to show \Validity\ and \UnikAg. Henceforth, 
let $i$ be a (possibly faulty) process that decides in some run of $\UOptMink$, let $m_i$ be the decision time of $i$, and
let $v$ be the value upon which $i$ decides.
Thus, there exists $m'_i \in \{m_i,m_i-1\}$ s.t.\ $\node{i,m'_i}$ is low or has hidden capacity $<k$, and s.t.\ $v = \minval{i,m'_i}$.
(To show this when $m_i=\bigl\lfloor\frac{t}{k}\bigr\rfloor+1$, we note that in this case $m_i>\bigl\lfloor\frac{f}{k}\bigr\rfloor$, and so, as shown in the stopping-time analysis above, this implies that $\node{i,m_i}$ has hidden capacity $<k$.)

\Validity:
As $v=\minval{i,m'_i}$, we have $K_i\exists v$ at $m'_i$, and thus $\exists v$.

\UnikAg:
It is enough to show that at most $k-1$ distinct values
smaller than $v$ are decided upon in the current run.
If $\node{i,m'_i}$ is low, then $v = \minval{i,m'_i} < k-1$, and thus there do not exist more than $k-1$ distinct legal values smaller than $v$, let alone ones decided
upon. For the rest of this proof we assume, therefore, that $\node{i,m'_i}$ is high, and so has hidden capacity $<k$.

Let $w<v$ be a value decided upon by some process. Let $j$ be this process, and let~$m_j$ be the time at which $j$ decides on $w$.
Thus, $w = \minval{i,m'_j}$ for some $m_j' \in \{m_j,m_j-1\}$ s.t.\ if $m'_j=m_j$, then either $K_j\cw$ at $m_j$, or $m_j=\bigl\lfloor\frac{t}{k}\bigr\rfloor+1$.

We first show that $m'_j \ge m'_i$. If $m'_j=m_j$ and $m_j=\bigl\lfloor\frac{t}{k}\bigr\rfloor+1$, then we immediately have
$m'_j=\bigl\lfloor\frac{t}{k}\bigr\rfloor+1\ge m_i \ge m'_i$, as required. Otherwise, the analysis is somewhat more subtle. We
first show that in this case, if $i$ is active at $m'_j+1$, then $K_i\exists w$ at $m'_j+1$. We reason by cases, according to the value of $m'_j$.
\begin{itemize}
\item
If $m'_j=m_j$, then $K_j\cw$ at $m'_j$, and thus there exists a process $k$ that never fails,
s.t.\ $K_k\exists w$ at $m'_j$. As $k$ never fails, $\node{k,m'_j}$ is seen by $\node{i,m'_j+1}$, and thus $K_i\exists w$ at $m'_j+1$, as required.
\item
Otherwise, $m'_j=m_j-1$. As $j$ is active at $m_j$, it does does not fail at $m'_j<m_j$, and therefore $\node{j,m'_j}$ is seen by
$\node{i,m'_j+1}$. Thus, as $K_j\exists w$ at $m'_j$, we obtain that $K_i\exists w$ at $m'_j+1$ in this case as well.
\end{itemize}
As $w<v$ and as $v = \minval{i,m'_i}$, we have $\lnot K_i\exists w$ at $m'_i$. Thus, we obtain that $m'_i < m'_j+1$, and therefore
$m'_j \ge m'_i$ in this case as well, as required. We have thus shown that we always have $m'_j \ge m'_i$.

As $\node{i,m'_i}$ does not have hidden capacity $k$, there exists $0\le\ell\le m'_i$ s.t.\ no more than $k-1$ processes at time $\ell$ are hidden from~$\node{i,m'_i}$.
As $m'_i \ge \ell$, we have $m'_j \ge m'_i \ge \ell$.
Let $H$ be the set of all processes seen at $\ell$ by $\node{j,m'_j}$. (Note
that if $m'_j = \ell$, then $H = \set{j}$.)
Since $m'_j \ge \ell$,
we have $\knownvals{j,m'_j} = \bigcup_{h \in H} \knownvals{h,\ell}$. Thus, $w=\minval{j,m'_j}=\min_{h \in H}\{\minval{h,\ell}\}$. Therefore,
$w = \minval{h,\ell}$ for some $h \in H$. As $\lnot K_i \exists w$ at $m'_i$, we thus have that $\node{h,\ell}$ is not seen
by $\node{i,m'_i}$. As $\node{h,\ell}$ is seen by $\node{j,m'_j}$, $h$ does not fail
before $\ell$, and thus $\node{h,\ell}$ is hidden from $\node{i,m'_i}$.
To conclude, we have shown that
\[w \in \bigl\{ \minval{h,\ell} \mid
\mbox{$\node{h,\ell}$ is hidden from $\node{i,m'_i}$} \bigr\}.\]
As there are at most $k-1$ processes hidden at $\ell$ from $\node{i,m'_i}$,
we conclude that no more than $k-1$ distinct values lower than $v$ are
decided upon, and the proof is complete.
\end{proof}

\section{Different Types of Unbeatability}
\label{sec-notions}

We first formally define last-decider unbeatability.

\begin{definition}[Last-Decider Domination and Unbeatability]
\leavevmode
\begin{itemize}
\item
A decision protocol $Q$ \defemph{last-decider dominates} a protocol~$P$ in~$\gamma$, denoted by $Q\boldsymbol{\overset{\smash{l.d.}}{\dom}_\gamma} P$ if, for all adversaries $\alpha$, if $i$ the last decision in~$P[\alpha]$ is at time $m_i$, then all decisions in $Q[\alpha]$ are taken before or at $m_i$. Moreover, we say that $Q$  \defemph{strictly last-decider dominates} $P$
if $Q\overset{\smash{l.d.}}{\dom}_\gamma P$ and  $P\!\!\boldsymbol{\not}\!\!\!\overset{\smash{l.d.}}{\dom}_\gamma Q$. I.e., if for some $\alpha\in\gamma$ the last decision in $Q[\alpha]$ is {\em strictly before} the last decision in $P[\alpha]$.
\item
A protocol $P$ is a \defemph{last-decider unbeatable} solution to a decision task~$S$ in a context~$\gamma$ if $P$ solves~$S$ in~$\gamma$ and no protocol $Q$ solving~$S$ in~$\gamma$ strictly last-decider dominates~$P$.%
\end{itemize}
\end{definition}

\begin{remark}
\leavevmode
\begin{itemize}
\item
If $Q\boldsymbol{\dom_\gamma} P$, then $Q\boldsymbol{\overset{\smash{l.d.}}{\dom}_\gamma} P$. (But not the other way around.)
\item
None of the above forms of strict domination implies the other.
\item
None of the above forms of unbeatability implies the other.
\end{itemize}
\end{remark}

Last-decider domination does not imply domination in the sense of the rest of this paper (on which our proofs is based). 
Nonetheless, the specific property of protocols dominating $\OptZ$, $\OptMaj$, $\OptMink$ and $\UOptZ$, which we use to prove that these protocols are unbeatable, holds also for protocols that only last-decider dominate these protocols. 

\begin{lemma}\label{last-dom-sufficient}
\leavevmode
\begin{parts}
\item\label{last-dom-sufficient-pz}
Let $Q\overset{\smash{l.d.}}{\dom}\Pz$ satisfy \Decision. If $K_i\exists0$ at $m$ in a run $r\!=\!Q[\alpha]$ of $Q$, then $i$ decides in $r$ no later than at $m$.
\item\label{last-dom-sufficient-maj}
Let $Q\overset{\smash{l.d.}}{\dom}\OptMaj$ satisfy Decision. If $K_i(\Maj=v)$ for $v \in \{0,1\}$ at $m$ in a run $r\!=\!Q[\alpha]$ of $Q$, then $i$ decides in $r$ no later than at $m$.
\item\label{last-dom-sufficient-mink}
Let $Q\overset{\smash{l.d.}}{\dom}\OptMink$ satisfy \Decision. If $i$ is low at $m$ in a run $r\!=\!Q[\alpha]$ of $Q$, then $i$ decides in $r$ no later than at $m$.
\item\label{last-dom-sufficient-upz}
Let $Q\overset{\smash{l.d.}}{\dom}\UPz$ satisfy \Decision. If $K_i\cz$ at $m$ in a run $r\!=\!Q[\alpha]$ of $Q$, then $i$ decides in $r$ no later than at $m$.
\end{parts}
\end{lemma}

The main idea in the proof of each of the parts of \cref{last-dom-sufficient} is to show that $i$ considers it possible that
all other active processes also know the fact stated in that part,
and so they must all decide by the current time in the corresponding run of the dominated protocol. Hence, the last decision decision in that run is made in the current time; thus, by last-decider domination, $i$ must decide.
The proofs for the first three parts are somewhat easier, as in each of these parts, any process at $m$ who sees (at least) the nodes seen by $\node{i,m}$
(or has the same initial value, if $m\!=\!0$) also knows the relevant
fact stated in that part. We demonstrate this by proving \cref{last-dom-sufficient-pz}; the analogous proofs of \cref{last-dom-sufficient-maj,last-dom-sufficient-mink} are left to the reader.

\begin{proof}[Proof of \cref{last-dom-sufficient-pz} of \cref{last-dom-sufficient}]
If $m\!=\!0$, then there exists a run $r'\!=\!Q[\beta]$ of $Q$, s.t.~~\emph{i)} $r'_i(0)\!=\!r_i(0)$,~~\emph{ii)} in $r'$ all initial values are $0$, and~~\emph{iii)} $i$ never fails in $r'$. Hence, in $\Pz[\beta]$ all decisions
are taken at time $m\!=\!0$, and therefore so is the last decision. Therefore, the last decision in $r'$ must be taken at time $0$. As $i$ never fails in $r'$, by \Decision\ it must decide at some
point during this run, and therefore must decide at $0$ in $r'$. As $r_i(0)\!=\!r'_i(0)$, $i$ decides at $0$ in $r$ as well, as required.

If $m\!>\!0$, then there exists a process $j$ s.t.\ $K_j\exists0$ at $m-1$ in $r$ and $\node{j,m-1}$ is seen by $\node{i,m}$. Thus, there exists a run $r'\!=\!Q[\beta]$ of $Q$,
s.t.~~\emph{i)} $r'_i(m)\!=\!r_i(m)$, and~~\emph{ii)} $i$ and $j$ never fail in $r'$. Thus, all processes that are active at $m$ in $r'$ see $\node{j,m-1}$ in $r'$ and therefore know $\exists0$ in $r'$.
Hence, in $\Pz[\beta]$ all decisions are taken by time $m$, and therefore so is the last decision. Therefore, the last decision in $r'$ must be taken no later than at time $m$.
As $i$ never fails in $r'$, by \Decision\ it must decide at some point during this run, and therefore must decide by $m$ in $r'$. As $r_i(m)\!=\!r'_i(m)$, $i$ decides by $m$ in $r$ as well, as required.
\end{proof}

As the proof of \cref{last-dom-sufficient-upz} is slightly more involved, we show it as well.

\begin{proof}[Proof of \cref{last-dom-sufficient-upz} of \cref{last-dom-sufficient}]
If $m\!=\!0$, then by \cref{lem:u-know}, $t\!=\!0$. There exists a run $r'\!=\!Q[\beta]$ of $Q$, s.t.~~\emph{i)} $r'_i(0)=r_i(0)$, and~~\emph{ii)} in $r'$ all initial values are $0$. Therefore,
as $t\!=\!0$, we have by \cref{lem:u-know} that all processes know $\cz$ at $m\!=\!0$ in $r'$. Hence, in $\UPz[\beta]$ all decisions
are taken at time $m\!=\!0$, and therefore so is the last decision. Therefore, the last decision in $r'$ must be taken at time $0$ as well. Since $t\!=\!0$, $i$ never fails in $r'$, and so by \Decision\ it must decide at some
point during this run, and therefore must decide at $0$ in $r'$. As $r_i(0)\!=\!r'_i(0)$, $i$ decides at $0$ in $r$ as well, as required.

If $m\!>\!0$, then there exists a process $j$ s.t.\ $K_j\exists0$ at $m\!-\!1$ in $r$ and $\node{j,m-1}$ is seen by $\node{i,m}$ in $r$.
Furthermore, as $t\!<\!n$, there exists a set of processes $I$
s.t.~~\emph{i)} $i,j \notin I$,~~\emph{ii)} $|I|=t\!-\!\knownf{i,m}\!-\!1$, and~~\emph{iii)} $\node{k,m\!-\!1}$ is seen by $\node{i,m}$ for every $k \in I$. Thus, there exists a run $r'=Q[\beta]$ of $Q$, s.t.~~\emph{i)} $r'_i(m)\!=\!r_i(m)$,~~\emph{ii)} $i$ and $j$ never fail in $r'$,~~\emph{iii)} all of $I$ fail in $r'$ at $m\!-\!1$, successfully sending messages only to $i$, and~~\emph{iv)} every process at $m\!-\!1$ in $r'$ that is not seen by $\node{i,m}$, is not seen by any other process at $m$ as well. We henceforth reason about $r'$. Every process $k \ne j$ that is active at $m$ sees $\node{j,m\!-\!1}$ and furthermore satisfies
$\knownf{k,m}\ge\knownf{i,m}+|I|=t-1$. Thus, by \cref{lem:u-know}, $K_k\cz$ at $m$, and thus $k$ decides at $(\UPz[\beta],m)$.
Additionally, as $K_j\exists0$ at $m\!-\!1$, by \cref{lem:u-know} $K_j\cz$ at $m$, and thus $j$ decides at $(\UPz[\beta],m)$.
Hence, in $\UPz[\beta]$ all decisions are taken by time $m$, and therefore so is the last decision. Therefore, the last decision in $r'$ must be taken no later than at time $m$.
As $i$ never fails in $r'$, by \Decision\ it must decide at some point during this run, and therefore must decide by $m$ in $r'$. As $r_i(m)=r'_i(m)$, $i$ decides by $m$ in $r$ as well, as required.
\end{proof}

As explained above, \cref{thm:last-decider} follows from \cref{last-dom-sufficient}, and from the proofs of \cref{thm:optz,thm:optmaj,thm:optmink,thm:u-opt}.

\mbox{~}\\
Finally, we sketch the structure of communication-efficient implementations for the protocols proposed in the paper:
\begin{lemma}
\label{nlogn}
For each of the protocols $\OptZ$, $\OptMaj$, $\OptMink$, $\UOptZ$ and $\UOptMink$ there is a protocol with identical decision times for all adversaries, in which every process sends at most $O(n\log n)$ bits overall to each other process. 
\end{lemma}

\begin{proof}(Sketch)
Moses and Tuttle in \cite{MT} show how to implement full-information protocols in the crash failure model with linear-size messages. In our case, a further improvement is possible, since decisions in all of the protocols depend only on the identity of hidden nodes and on the vector of initial values. In a straightforward implementation, we can have a process $i$ report  ``{\tt value}$(j) = v$'' once for every $j$ whose initial value it discovers, and ``{\tt failed\_at}$(j) = \ell$'' once where~$\ell$ is the earliest failure round it knows for $j$. In addition, it should send an ``{\tt I'm\_alive}'' message in every round in which it has nothing to report. Process~$i$ can send at most one {\tt value} message and two 
{\tt failed\_at} messages for every $j$. Since {\tt I'm\_alive} is a constant-size message sent fewer than~$n$ times, and since encoding~$j$'s ID requires $\log n$ bits, a process~$i$ sends a total of $O(n \log n)$ bits overall. 
\end{proof}

\end{document}